\documentclass[aps,prl,superscriptaddress,twocolumn,showpacs,titlepage=false,longbibliography]{revtex4-1}
\usepackage{graphicx}
\usepackage{amsmath}
\usepackage{amsmath}
\usepackage{mathrsfs}
\usepackage{dsfont}
\usepackage{amssymb}
\usepackage{subfigure}
\usepackage{color}
\usepackage{blkarray}
\usepackage[breaklinks=true,colorlinks=true,linkcolor=blue,urlcolor=blue,citecolor=blue]{hyperref}
\def\bra#1{{\left\langle #1 \right|}}
\def\ket#1{{\left| #1 \right\rangle}}

\definecolor{amber}{rgb}{1.0, 0.49, 0.0}
\usepackage[dvipsnames]{xcolor}

\newcommand{\bs}[1]{\boldsymbol{#1}}

\usepackage{verbatim}
\usepackage{amsthm}
\usepackage{enumerate}
\usepackage{bbm}
\usepackage{hyperref}

\newtheorem{theorem}{Theorem}

\newtheorem{proposition}{Proposition}
\newtheorem*{theorem2}{Theorem}

\begin{document}

\title{Multi-parameter estimation in networked quantum sensors}
\author{Timothy J. Proctor}
\thanks{The first two authors contributed equally to this work. Email address for TJP: tjproct@sandia.gov. Email address for PAK: Paul.Knott@nottingham.ac.uk.}
\affiliation{Sandia National Laboratories, Livermore, CA 94550, USA}
\affiliation{Department of Chemistry, University of California, Berkeley, CA 94720, USA}
\author{Paul A. Knott}
\thanks{The first two authors contributed equally to this work. Email address for TJP: tjproct@sandia.gov. Email address for PAK: Paul.Knott@nottingham.ac.uk.}
\affiliation{Centre for the Mathematics and Theoretical Physics of Quantum Non-Equilibrium Systems (CQNE), School of Mathematical Sciences, University of Nottingham, University Park, Nottingham NG7 2RD, UK}
\affiliation{Department of Physics and Astronomy, University of Sussex, Brighton BN1 9QH, UK}
\author{Jacob A. Dunningham}
\affiliation{Department of Physics and Astronomy, University of Sussex, Brighton BN1 9QH, UK}
\date{\today}

\begin{abstract}
We introduce a general model for a network of quantum sensors, and we use this model to consider the question: When can entanglement between the sensors, and/or global measurements, enhance the precision with which the network can measure a set of unknown parameters? We rigorously answer this question by presenting precise theorems proving that for a broad class of problems there is, at most, a very limited intrinsic advantage to using entangled states or global measurements. Moreover, for many estimation problems separable states and local measurements are optimal, and can achieve the ultimate quantum limit on the estimation uncertainty. This immediately implies that there are broad conditions under which simultaneous estimation of multiple parameters cannot outperform individual, independent estimations. Our results apply to any situation in which spatially localized sensors are unitarily encoded with independent parameters, such as when estimating multiple linear or non-linear optical phase shifts in quantum imaging, or when mapping out the spatial profile of an unknown magnetic field. We conclude by showing that entangling the sensors \emph{can} enhance the estimation precision when the parameters of interest are global properties of the entire network.
\end{abstract}
\maketitle

Quantum networks are central to a growing number of quantum information technologies, including quantum computation \cite{kimble2008quantum,nickerson2013topological} and cryptography \cite{Sasaki2011field,wang2013direct}. Many important metrology problems can be framed in terms of networks, including mapping magnetic fields \cite{steinert2010high,hall2012high,pham2011magnetic,seo2007fourier,baumgratz2016quantum}, phase imaging \cite{humphreys2013quantum,liu2014quantum,yue2014quantum,ciampini2015quantum,knott2016local,gagatsos2016gaussian,zhang2017quantum} and global frequency standards \cite{komar2014quantum}. However, there is no general consensus on whether entanglement within a network of sensors can enhance the precision to which the network can measure a set of unknown parameters: entanglement provides significant enhancements in some cases \cite{komar2014quantum,eldredge2016optimal} but not others \cite{knott2016local,kok2017role}. Given the immense challenges faced in the creation and manipulation of entangled states, developing a complete understanding of when such resources are advantageous for multi-parameter estimation is of paramount importance.

In this letter we introduce and analyze a general model that encompasses a wide range of those quantum multi-parameter estimation (MPE) problems that might naturally be termed a ``quantum sensing network'' (QSN). Our QSN model (Fig.~\ref{fig:quantum-sensing-networks}) includes any situation in which spatially or temporally localized sensors are encoded with independent parameters. Hence, our results have direct implications for multi-mode linear \cite{humphreys2013quantum,liu2014quantum,yue2014quantum,ciampini2015quantum,zhang2017quantum,knott2016local,gagatsos2016gaussian} or non-linear \cite{liu2014quantum} optical phase shift estimation for quantum imaging, mapping unknown spatially or temporally changing fields \cite{steinert2010high,hall2012high,pham2011magnetic,seo2007fourier,baumgratz2016quantum}, estimating many-qubit Hamiltonians \cite{eldredge2016optimal}, and networks comprised of clocks \cite{komar2014quantum}, BECs \cite{pyrkov2013entanglement},  interferometers \cite{knott2016local}, or hybrid elements \cite{wallquist2009hybrid}. Beyond these examples, any situation in which independent parameters are unitarily imprinted on different quantum subsystems fits into our model.

\begin{figure}
\includegraphics[width=7cm]{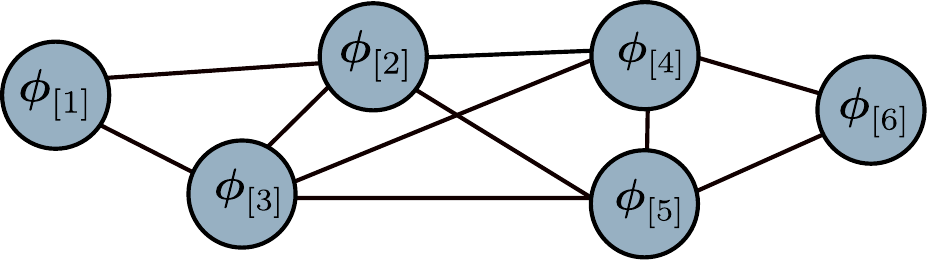}
\caption{A network of quantum sensors. The $k^{\rm th}$ node represents a ``sensor'' into which the vector parameter $\bs{\phi}_{[k]}$ is encoded via a local unitary evolution. The connections between the nodes denote that, in general, the sensors can be entangled, and/or global measurements can be performed.}
\label{fig:quantum-sensing-networks}
\end{figure}

Using our model we show that, if the generators of all of the unknown parameters commute, no fundamental precision enhancement can be achieved by entangling the sensors or by performing global measurements. In this case, states that are separable between the sensors -- which are often easier to prepare experimentally -- can achieve the ultimate quantum limit.  We then look at the case of non-commuting parameter generators; here we demonstrate that entanglement between sensors can \emph{at most} enhance the estimation precision by a factor of two. We conclude by showing that entangling the sensors \emph{can} significantly enhance the precision when estimating global parameters, such as the average of all the unknown parameters in the network \cite{komar2014quantum}.

Whenever a protocol employs entangled resources it is fundamentally indivisible into separate, independent estimations at each location: it is intrinsically a \emph{simultaneous} \cite{humphreys2013quantum,liu2014quantum,yue2014quantum,ciampini2015quantum,knott2016local,gagatsos2016gaussian,zhang2017quantum,yue2014quantum,baumgratz2016quantum,szczykulska2016multi} estimation method. As such, our results directly imply that there are broad conditions under which simultaneous estimation cannot outperform a strategy that estimates each parameter individually, conclusively proving that enhancements from simultaneous estimation \cite{humphreys2013quantum,liu2014quantum,yue2014quantum,ciampini2015quantum,zhang2017quantum,yue2014quantum,baumgratz2016quantum,szczykulska2016multi} are not generic.

\vspace{0.1cm}
\noindent{\bf Multi-parameter estimation (MPE) --} Consider a quantum system with Hilbert space $\mathcal{H}$, and let $\mathscr{D}(\mathcal{H})$ and $\mathscr{M}(\mathcal{H})$ denote the space of density operators and positive-operator valued measures (POVMs) on $\mathcal{H}$, respectively. We will use the standard framework for a quantum metrology protocol \cite{giovannetti2011advances,zwierz2010general,demkowicz2012elusive}: An experimenter picks some $\rho \in \mathscr{D}(\mathcal{H})$ and $\mathcal{M} \in \mathscr{M}(\mathcal{H})$ and implements $\mu$ repeats of: i) prepare $\rho$; ii) let $\rho$ evolve to $\rho_{\bs{\phi}} = U_{\bs{\phi}} \rho U_{\bs{\phi}}^{\dagger}$ where $U_{\bs{\phi}}$ is a unitary that depends on $d$ unknown parameters $\boldsymbol{\phi} = (\phi_1, \phi_2, \dots, \phi_d)^T$; iii) apply the measurement  $\mathcal{M}$ to $\rho_{\bs{\phi}}$. An estimate of $\bs{\phi}$ is then calculated from experimental outcomes using an estimator $\bs{\Phi}$.

A common measure of the estimation uncertainty is the covariance matrix $\text{Cov}(\bs{\Phi}) = \mathbb{E}[\left(\bs{\Phi} -  \mathbb{E}[ \bs{\Phi} ]\right)\left(\bs{\Phi} -  \mathbb{E}[\bs{\Phi}]\right)^T]$, where $\mathbb E[\cdot]$ is the expected value. For any unbiased estimator, the \emph{quantum Cram\'er-Rao bound} (QCRB) states that $\text{Cov}(\bs{\Phi}) \geq (\mathcal{F} \mu)^{-1}$ \cite{fujiwara1995quantum,matsumoto2002new,helstrom1976quantum,paris2009quantum,helstrom1967minimum}, where $\mathcal{F}$ is the \emph{quantum Fisher information matrix} (QFIM) for $\rho_{\bs{\phi}}$, defined by $\mathcal{F}_{kl} := \text{Tr} [ \rho_{\bs{\phi}} \hat{L}_{k} \hat{L}_{l} +  \rho_{\bs{\phi}} \hat{L}_{l} \hat{L}_{k}] / 2$ with $\hat{L}_{k}$ solving $\partial \rho_{\bs{\phi}} / \partial \phi_{k}  = (\rho_{\bs{\phi}} \hat{L}_{k} + \hat{L}_{k} \rho_{\bs{\phi}}) / 2$ \cite{fujiwara1995quantum,matsumoto2002new,helstrom1976quantum,paris2009quantum,helstrom1967minimum}. 
Note that for matrices $A$ and $B$, $A \geq B$ denotes that $A - B$ is positive semi-definite. For $d = 1$ and any $\rho_{\bs \phi}$ there is always a measurement and an estimator that saturate the QCRB as $\mu \to \infty$ \cite{braunstein1994statistical,paris2009quantum,demkowicz2014quantum}, but for $d > 1$ this is not generally true \cite{fujiwara1995quantum,vidrighin2014joint,helstrom1976quantum,fujiwara2001estimation,ragy2015resources,ragy2016compatibility,pezze2017optimal}. Some elements of $\bs{\phi}$ may be of more interest than others, so we introduce a $d\times d$ diagonal \emph{weighting matrix}, $W$, with $W \geq 0$, and define the scalar quantity $E_{\bs{\Phi}} := \text{Tr}(W\text{Cov}(\bs{\Phi}))$ \cite{genoni2013optimal,vaneph2013quantum,fujiwara1995quantum}. Throughout this letter, $E_{\bs{\Phi}}$ is the figure of merit to minimize. The QCRB implies that $E_{\bs{\Phi}} \geq  \frac{1}{\mu} \sum_k W_{kk} [\mathcal{F}^{-1}]_{kk}$.

\vspace{0.1cm}
\noindent{\bf Quantum sensing networks --} In this letter we consider a particular class of quantum MPE problems: \emph{quantum sensing networks} (QSNs). A QSN is, by definition, any estimation problem in which we have $s$ quantum systems, which we will call ``quantum sensors'', and there are $d$ unknown parameters with each parameter unitarily encoded into \emph{one and only one} of the sensors.  It is natural to refer to this model as a QSN because any set of spatially distributed quantum systems that are each ``sensing'' some locally unitarily encoded parameters is a QSN (although some systems without this spatial structure also fit into this framework). 

Our model, illustrated in Fig.~\ref{fig:quantum-sensing-networks}, encompasses many metrology problems in the literature \cite{humphreys2013quantum,liu2014quantum,yue2014quantum,ciampini2015quantum,zhang2017quantum,knott2016local,gagatsos2016gaussian,steinert2010high,hall2012high,pham2011magnetic,seo2007fourier,baumgratz2016quantum,eldredge2016optimal, komar2014quantum,pyrkov2013entanglement,wallquist2009hybrid} (see examples later). More formally, a QSN is any MPE problem in which the total Hilbert space $\mathcal{H}$ may be decomposed as $\mathcal{H} = \mathcal{H}_1  \otimes \cdots \otimes \mathcal{H}_{s}$ for some $\{ \mathcal{H}_k \}$, and the unitary evolution may be decomposed as $U_{\bs{\phi}} = U_1(\bs{\phi}_{[1]}) \otimes U_2(\bs{\phi}_{[2]}) \otimes \cdots \otimes U_s(\bs{\phi}_{[s]})$, where $\bs{\phi}_{[k]}$ denotes the $d_k$-dimensional sub-vector of $\bs{\phi}$ encoded onto the $k^{\text{th}}$ sensor by the unitary $U_k$, with $\sum_k d_k = d$. Let $\bs{\phi}_{[1]}=(\phi_1, \dots,\phi_{d_1})^T$, $\bs{\phi}_{[2]}=(\phi_{d_1+1}, \dots,\phi_{d_1+d_2})^T$, etc.

Often we wish to compare probe states $\rho$ that contain the same quantity of ``resources'' $R(\rho)$, for some $R : \mathscr{D}(\mathcal{H}) \to \mathbb{R}_{\geq 0}$. There is no universally applicable definition for the resources within a state; we will consider functions of the form $R(\rho) = \text{Tr}[( \hat{R}_1 + \hat{R}_2 + \cdots + \hat{R}_{s} )(\rho)]$,
where $\hat{R}_k$ is any Hermitian operator acting non-trivially only on sensor $k$ and satisfying $R(\rho_{\bs{\phi}}) = R(\rho)$ (so resources are conserved under the evolution). This includes the resource counting in most standard metrology problems. E.g., in optical metrology with $s$ modes the total average number of photons is the standard resource \cite{humphreys2013quantum,liu2014quantum,yue2014quantum,ciampini2015quantum,knott2016local,gagatsos2016gaussian}, given by $\hat{R}_k = \hat{n}_k$, where $\hat{n}_k$ is the number operator on mode $k$ (which commutes with the standard parameter generator, $\hat{n}_k$). In atomic sensing, the resource is normally the total number of atoms \cite{mpe_prl_note1,huelga1997improvement,tanaka2014robust,kessler2014quantum}. This is obtained by taking the Hilbert space of each sensor to be the direct sum of the $n$-atoms Hilbert space for $n=0,1,2,\dots$, and $\hat{R}_k$ to be the atom-counting operator, which commutes with all atom-number conserving Hamiltonians.

\vspace{0.1cm}
\noindent {\bf QSNs with commuting parameter generators --} The \emph{generator} of $\phi_k$ is defined by $\hat{H}_k := -i (\partial U_{\bs{\phi}}^{\dagger} / \partial_{\phi_k}) U_{\bs{\phi}}$ \cite{liu2015quantum,liu2014fidelity}. Our main results are separated into two cases: when  the generators all commute, and when they do not. First, consider any QSN in which the generators all commute. Informally, our first result is that for \emph{any} such estimation problem sensor-separable states can enable an estimation uncertainty that is at least as small as can be achieved with sensor-entangled states. This also implies that, in this setting, simultaneous estimation provides no intrinsic advantage over individual estimation; the latter can achieve the ultimate quantum limit. We now state this precisely:

\begin{theorem}
Consider any QSN in which $[\hat{H}_k,\hat{H}_l] = 0$ for all $k,l$ and where we wish to minimize $E_{\bs{\Phi}}$ where $E_{\bs{\Phi}} = \rm{Tr}( W\rm{Cov}( \bs{\Phi} ) )$ for some specified $W$. For any estimator, probe $\rho$ and measurement $\mathcal{M}_{\rho}$, there exists an estimator, a probe $\varphi$ and a measurement $\mathcal{M}_{\varphi}$ for which 
\begin{enumerate}
\item $\varphi$ is separable between sensors.
\item $R(\varphi) \leq R(\rho)$.
\item $\mathcal{M}_{\varphi}$ is implementable by independent measurements of each sensor.
\item $E_{\bs{\Phi}}(\varphi,\mathcal{M}_{\varphi}) \leq E_{\bs{\Phi}}(\rho,\mathcal{M}_{\rho})$ in the asymptotic $\mu$ limit.
\end{enumerate}
\end{theorem}

\begin{proof}
This may be proven by constructing such a $\varphi$ and $\mathcal{M}_{\varphi}$, for arbitrary $\rho$ and $\mathcal{M}_{\rho}$. First consider pure $\rho$, i.e., $\rho = \psi = \ket{\psi}\bra{\psi}$. We now find a mapping from $\psi$ to a state $\varphi$ that satisfies conditions 1 and 2, and that has an equal or smaller QCRB on $E_{\bs{\Phi}}$. Consider the state $\ket{\varphi} = \bigotimes_{k=1}^s (\sum_{\lambda_k} \|\langle \psi | \lambda_k \rangle\| \ket{\lambda_k} )$, where $\{\ket{\lambda_k}\}$ is a set of orthonormal mutual eigenstates of the generators for all of the parameters encoded into sensor $k$. By construction, $\psi$ and $\varphi$ have the same statistics for any operator that is diagonal in the eigenbasis of the generators, and $\varphi$ is separable between sensors. As the resource operator commutes with $U_{\bs{\phi}}$, it commutes with the parameter generators, implying $\varphi$ satisfies conditions 1 and 2.

For a pure state and commuting generators $\mathcal{F}_{kl} = 4(\langle \hat{H}_k\hat{H}_l \rangle - \langle \hat{H}_k \rangle \langle \hat{H}_l  \rangle)$ \cite{knott2016local,baumgratz2016quantum,yue2014quantum}. Using this we find that $\psi$ and $\varphi$ have the same block-diagonal QFIM elements, where the block diagonals are the sub-QFIMs for each $\bs{\phi}_{[k]}$, denoted $\mathcal{F}_{[kk]}$, and $\varphi$ has a block-diagonal QFIM ($\psi$ in general does not). Now for any QFIM $[\mathcal{F}^{-1}]_{[kk]} \geq [\mathcal{F}_{[kk]}]^{-1}$, with saturation only for a block-diagonal QFIM (see the appendix), and hence the diagonal elements of the inverse QFIM of $\varphi$ are all smaller than or equal to those of $\psi$. Using $E_{\bs{\Phi}} \geq  \frac{1}{\mu} \sum_l W_{ll} [\mathcal{F}^{-1}]_{ll}$, and noting that when the generators commute there always exists a measurement and estimator that asymptotically saturate the QCRB \cite{matsumoto2002new}, we see that condition 4 is satisfied by some measurement and estimator. It only remains to show that for one such measurement condition 3 holds, and for every mixed state $\rho$, there exists a pure state with equal or lower $E_{\bs{\Phi}}$ and the same resources. We prove this in the appendix.
\end{proof}

Theorem 1 has practical implications for a range of important estimation problems. For example, consider estimating a set of $d$ optical phases encoded into $d$ modes (defined with respect to a classical phase reference \cite{jarzyna2012quantum}). Theorem 1 implies that, for any mode-entangled state and measurement, there is a mode-separable state and measurement (acting on only that mode and a local phase reference) that provides an equal or lower estimation uncertainty, for the same average number of photons through the $d$ phase shifts. So, although highly mode-entangled states can provide high estimation precision \cite{humphreys2013quantum,liu2014quantum,yue2014quantum,zhang2017quantum}, this entanglement is not necessary. This supersedes the results of Ref.~\cite{knott2016local}, which apply only to mode-symmetric states.

Importantly, Theorem 1 is only directly applicable when the set of states, from which we wish to find the best $\rho$, is the set of all density operators on $\mathcal{H}=\mathcal{H}_1\otimes\dots \otimes\mathcal{H}_s$. Hence, if we restrict the allowed $\rho$ to $\mathbb{S} \subset \mathscr{D}(\mathcal{H})$, Theorem 1 is only applicable if $\mathbb{S}$ contains all $\rho$ on some smaller Hilbert space $\mathcal{H}'$ that still factorizes. This is not the case for some global constraints on the state. This reconciles our theorem with Humphreys \emph{et al.}~\cite{humphreys2013quantum}, who show that highly-entangled ``generalized NOON states'' provide a precision enhancement over individual estimation strategies, for the $d$-optical-phases problem, when only states with definite total photon number are considered. 


Interestingly, Theorem 1 may be extended to further classes of $\mathbb{S}$. This includes any $\mathbb{S}$ containing pure states whereby every state in $\mathbb{S}$ can be mapped to a sensor-separable state in $\mathbb{S}$ with the same measurement statistics for operators diagonal in the eigenbasis of the generators (the proof is a trivial adaption of that given above). This implies that, if considering only Gaussian optical states in the $d$-phases problem, entanglement cannot reduce the estimation uncertainty. As such, our theorem strengthens and complements the results of Ref.~\cite{gagatsos2016gaussian}.

Theorem 1 may also be applied to other important metrology scenarios: It implies that the estimation of non-linear optical phase shifts on many modes \cite{liu2014quantum} does not benefit from mode-entanglement, and in a network of clocks \cite{komar2014quantum}, if each clock is used for local timekeeping then entangling the clocks will not enhance the precision. A magnetic field sensing problem is considered later. 

\vspace{0.1cm}
\noindent
{\bf QSNs with non-commuting parameter generators --} There are a variety of important estimation problems for which the generators do \emph{not} commute \cite{zhuang2017entanglement,baumgratz2016quantum,ballester2004estimation}, such as estimating the three spatial components of a magnetic field \cite{baumgratz2016quantum}, or estimating completely unknown unitaries \cite{ballester2004estimation}. We now adapt Theorem 1 to the case of non-commuting parameter generators. 

Consider an arbitrary QSN with some non-commuting parameter generators. In our model, the generators of parameters imprinted on different sensors always commute, so only the generators of parameters encoded into the same sensor can be non-commuting. When estimating parameters with non-commuting generators, it is known that the optimal estimation protocol will generally require a probe that is entangled with an ancilla \cite{ballester2004estimation,zhuang2017entanglement}. In a QSN, other sensors in the network can potentially play a similar role to ancillas, and so sensor-entanglement might reduce estimation uncertainty. However, any enhancement in the estimation precision gained from entanglement between sensors can instead be obtained by entangling each sensor with a local ancilla. The cost of this is that resources can be consumed by the ancillary system; twice the resources might be required to obtain the same estimation precision without sensor-entanglement. We can state this precisely in the following theorem:

\begin{theorem}
Consider any QSN in which we wish to minimize $E_{\bs{\Phi}}$. For any estimator, probe state $\rho \in \mathscr{D}(\mathcal{H})$ and measurement $\mathcal{M}_{\rho} \in \mathscr{M}(\mathcal{H})$, there exists an estimator, probe $\varphi \in \mathscr{D}(\mathcal{H} \otimes \mathcal{H})$ and measurement $\mathcal{M}_{\varphi} \in \mathscr{M}(\mathcal{H} \otimes \mathcal{H})$ for which 
\begin{enumerate}
\item $\varphi$ is separable between sensors, but each sensors can be entangled with a local ancilla.
\item $R(\varphi) \leq 2R(\rho)$.
 \item $\mathcal{M}_{\varphi}$ is implementable by independent measurements of each sensor.
\item $E_{\bs{\Phi}}(\varphi, \mathcal{M}_{\varphi}) \leq E_{\bs{\Phi}}(\rho, \mathcal{M}_{\rho})$ in the asymptotic $\mu$ limit.
\end{enumerate}
\end{theorem}
A complete proof is provided in the appendix (it closely follows the proof of Theorem 1). Note that condition 2 in this theorem depends on how resources used in ancillary sensors are counted, and here we have counted resources in the ancillas and sensors equally. If ancillas are considered cost-free then condition 2 improves to $R(\varphi) \leq R(\rho)$. Whether entanglement with a local ancilla is practically plausible is application dependent. 
Theorem 2 can be applied to a range of practical QSN problems. For example, if we wish to characterize a multi-dimensional field at multiple locations, then entanglement between atomic sensors at these locations can provide no improvement in precision compared to entangling these atoms with some local ancillary system (which may contribute to total resources used). This complements the results of Ref.~\cite{baumgratz2016quantum}, which provides strategies for single-site estimation of multi-dimensional fields.

\vspace{0.1cm}
\noindent
{\bf Estimating global functions of $\bs{\phi}$ --}  In some sensing problems it may not be necessary to estimate $\bs{\phi}$. Instead, the parameter(s) of interest could be some function(s) of $\bs{\phi}$, e.g., $\sum_k\phi_k$. In this case, the aim is to optimize the QSN for estimating these functions, and this encompasses many important problems, including measuring: phase differences in one \cite{aasi2013enhanced} or more \cite{knott2016local} interferometers; the average or sum of many parameters \cite{komar2014quantum}; a linear gradient \cite{zhang2014fitting,ng2014quantum}. A \emph{global} property of the network is some
vector (or scalar) with elements that are functions of $\{\phi_k\}$ depending non-trivially on many or all of the $\phi_k$, which includes the examples given above. We now show that the optimal protocol for estimating global properties of a QSN often requires sensor-entangled states.

For simplicity, we consider estimating a single linear function of $\bs{\phi}$; $\theta = \bs{v}^T \bs{\phi}$ for some $\bs{v} \in \mathbb{R}^d$. To fix arbitrary constants, let $\|\bs{v}\|_2 = 1$ and $v_k \geq 0$ $\forall k$ ($\| \bs{v} \|_{p} := [\sum_k |v_k|^p]^{1/p}$). Moreover, consider a QSN consisting of $\leq N$ particles (e.g., atoms or photons) distributed over $d$ sensors, with $\phi_{k}$ encoded into sensor $k$. We take the parameter generators to all be identical (except that they act on different sensors), with the maximal and minimal eigenvalues of the generator for $\leq n$ particles in a sensor, $\lambda_{\max,n}$ and $\lambda_{\min,n}$, satisfying $\lambda_{\max,n} - \lambda_{\min,n} = \kappa n$ for some constant $\kappa >0$. Denote corresponding orthonormal eigenvectors by $\ket{\lambda_{\max,n}}$ and $\ket{\lambda_{\min,n}}$. Examples that fit into this setting include estimating a function of many linear optical phase shifts, or of a spatially varying 1-dimensional magnetic field with multi-level atoms, or qubits \cite{eldredge2016optimal}.

Although we only wish to estimate $\theta$, there are many unknown parameters. Hence, to bound $\text{Var}(\Theta) = \mathbb{E}[\Theta^2] -  \mathbb{E}[\Theta]^2$ ($\Theta$ is the estimate of $\theta$) requires the QCRB on $\bs{\theta} = (\theta, \theta_2, \dots)^T = M\bs{\phi}$ for some matrix $M$ with $(M\bs{\phi})_1 = \theta$. We may take $M$ to be orthogonal, as only the first row of $M$ is specified by the problem. The relevant QFIM is then $\mathcal{F}(\bs{\theta}) = M \mathcal{F}(\bs{\phi}) M^T$ \cite{paris2009quantum}.

The optimal $n$-particle state of sensor $k$ for estimating $\phi_k$ is $\propto  \ket{\lambda_{\text{min},n}} +\ket{\lambda_{\max,n}}$, so the optimal $N$-particle QSN sensor-separable state for estimating $\theta$ is $ \propto \left(  \ket{\lambda_{\text{min},w_k}} +\ket{\lambda_{\max,w_k}} \right)^{\otimes d}$ optimized over $\bs{w} \in \mathbb{N}^d$ with $\|\bs{w}\|_1 = N$. By calculating the QFIM of this $\bs{w}$-optimized state, for any pure and sensor-separable state we have
 $\text{Var}(\Theta)  \geq   \| \bs{v} \|_{2/3}^{2}/(\mu \kappa^2 N^2)  \geq  \| \bs{v} \|_{1}^{3}/(\mu \kappa^2N^2) $, where $\mu$ is the number of experimental repeats. Now, assuming that $v_k / \|\bs{v}\|_1$ is rational and that $N$ is such that $\tilde{v}_{k} \equiv Nv_k/\|\bs{v}\|_1$ is an integer $\forall k$, consider the sensor-entangled GHZ-like state 
\begin{equation}
\ket{\psi_{\textsc{ghz},\bs{v}}} =\frac{1}{\sqrt{2}} \left( \ket{\lambda_{\text{max},\tilde{v}_{k}}}^{\otimes d}+   \ket{\lambda_{\text{min},\tilde{v}_{k}}}^{\otimes d} \right).
\end{equation}
 The QFIM for this state is $\mathcal{F}(\bs{\phi}) = \kappa^2N^2 \bs{v}\bs{v}^T / \|\bs{v}\|_1^2$, and hence $\mathcal{F}(\bs{\theta})_{11} = \kappa^2N^2 / \|\bs{v}\|_1^2$ with all other matrix elements zero. This QFIM is singular, but the state depends on $\theta$, so the saturable QCRB for this state is given by $\text{Var} (\Theta)  \geq 1 / (\mu\mathcal{F}(\bs{\theta})_{11}) = \|\bs{v}\|_1^2 / (\mu\kappa^2N^2)$. 

As  $\|\bs{v}\|_2=1$, for \emph{all} non-trivial $\bs{v}$ (i.e., $\bs{v}$ with multiple non-zero elements) $\|\bs{v}\|_1 >1$. Hence, for all such $\bs{v}$ entanglement between sensors reduces the estimation uncertainty below what is obtainable with \emph{any} sensor-separable state.  Moreover, $\|\bs{v}\|_1$ is maximal when $\bs{v} \propto (1,1,\dots,1)$, and so the precision enhancement is largest when estimating the average or sum of all $d$ parameters. In this setting, the reduction in the estimation variance is a factor of $1/d$ (as then $\| \bs{v} \|_1^2 / \| \bs{v}\|_{2/3}^2 = 1 / d$). 

To illustrate these results, we now apply them to a simple -- but practically relevant -- example: estimating the difference between the magnetic field strength at two locations with $N$ qubits (i.e., gradient estimation). Consider estimating $\theta = (\phi_2 - \phi_1)/\sqrt{2}$ with $\phi_k$ for $k=1,2$ generated by $\hat{J}_{z,k} = \frac{1}{2}\sum_j\sigma_{z,k,j}$ on sensor $k$, which consists of $n_k$ qubits for $n_1+n_2 = N$, where $\sigma_{z,k,j}$ is the $\sigma_{z}$ operator on qubit $j$ in sensor $k$. Our results imply that a global GHZ-like state $ \propto \ket{\downarrow}^{n_1}\ket{\uparrow}^{n_2} + \ket{\uparrow}^{n_1}\ket{\downarrow}^{n_2}$ with $n_1=n_2=N/2$ has an uncertainty reduction of $1/2$ compared to any sensor-separable state. However, if we instead wish to estimate $\phi_1$ and $\phi_2$ (or $\phi_2 - \phi_1$ and $\phi_2 + \phi_1$), then the above state is not appropriate, as it is sensitive only to $\phi_2 - \phi_1$. In this case, Theorem 1 implies that the optimal probe state is separable between the atoms at the two sites (the optimal state is then a local GHZ-like state at each site). Importantly, note that these conclusions do not necessarily hold if $\phi_1$ and $\phi_2$ have a known dependence: the extreme case is when we know that $\phi_1=\phi_2$, in which case estimating $\phi \equiv \phi_1 = \phi_2$ is a well-known one-parameter problem, and a global GHZ is optimal \cite{bollinger1996optimal,huelga1997improvement}. This example can be directly adapted to $l$-level atoms, $>2$ sensors, and more general linear functions.


Recently, Ge \emph{et al.}~\cite{ge2017distributed} have applied our results to the estimation of a function of $d$ linear phase shifts, and they have shown how to obtain the $O(d)$ precision enhancement, derived above, by entangling photons using a linear optical network. These interesting results show that the $O(d)$ enhancement proven here is potentially obtainable with current technology.

\vspace{0.1cm}
\noindent
{\bf Conclusions:} Quantum metrology is a powerful emerging technology, but while many practical problems unavoidably involve more than one unknown parameter, the critical resources for obtaining the ultimate quantum limit in multi-parameter estimation (MPE) are not yet well-understood. In this setting, simultaneous estimation, entanglement between sensors, and global measurements are possible avenues for improving estimation precision that are not relevant in the single-parameter scenario \cite{humphreys2013quantum,liu2014quantum,yue2014quantum,ciampini2015quantum,zhang2017quantum,yue2014quantum,baumgratz2016quantum,szczykulska2016multi}.

In this letter we considered a broad class of practically important MPE problems: \emph{quantum sensing networks}, meaning any setting in which the unknowns parameters can be sub-divided into distinct sets each associated with one spatially or temporally localized sensor. We have presented a general model for such estimation problems, and we stated precise theorems that show that simultaneous estimation, entanglement between sensors, and global measurements are broadly \emph{not} fundamentally useful resources for minimizing estimation uncertainty in this setting. The important exception to this is when one or more \emph{global} properties of the network are the parameters of interest, e.g., if only the average of all the parameters is to be estimated. In this case we have shown that entangled states and measurements can, in general, improve estimation precision. In doing so, we have shown that GHZ-like states have a particularly high precision for estimating generic linear functions in a practically relevant class of QSNs, including in optical and atomic sensing networks.

These results provide a rigorous foundation for understanding the role of entanglement and simultaneous estimation in optimal MPE, and they definitively show that these resources are not critical in a broad class of important problems. We anticipate that this letter will prove helpful for guiding the development of sensing technologies for multi-parameter metrology in fields as diverse as optical imaging \cite{humphreys2013quantum,liu2014quantum,yue2014quantum,ciampini2015quantum,zhang2017quantum,knott2016local,gagatsos2016gaussian}, field mapping with atoms \cite{steinert2010high,hall2012high,pham2011magnetic,seo2007fourier,baumgratz2016quantum}, and sensor networks comprised of BECs \cite{pyrkov2013entanglement}, clocks \cite{komar2014quantum}, or interferometers \cite{knott2016local}. Moreover, recently these results been applied to the interesting problem of estimating functions of linear optical phases \cite{ge2017distributed}.


We thank Jes{\'u}s Rubio for helpful discussions. This work was partly funded by the UK EPSRC through the Quantum Technology Hub: Networked Quantum Information Technology (grant reference EP/M013243/1), and the Foundational Questions Institute under the Physics of the Observer Programme (Grant No. FQXi-RFP-1601). Sandia National Laboratories is a multimission laboratory managed and operated by National Technology and Engineering Solutions of Sandia, LLC, a wholly owned subsidiary of Honeywell International, Inc., for the U.S. Department of Energy's National Nuclear Security Administration under contract DE-NA0003525. 

\bibliography{MyLib_Thesis_multiparameter_new}

\begin{thebibliography}{61}%
\makeatletter
\providecommand \@ifxundefined [1]{%
 \@ifx{#1\undefined}
}%
\providecommand \@ifnum [1]{%
 \ifnum #1\expandafter \@firstoftwo
 \else \expandafter \@secondoftwo
 \fi
}%
\providecommand \@ifx [1]{%
 \ifx #1\expandafter \@firstoftwo
 \else \expandafter \@secondoftwo
 \fi
}%
\providecommand \natexlab [1]{#1}%
\providecommand \enquote  [1]{``#1''}%
\providecommand \bibnamefont  [1]{#1}%
\providecommand \bibfnamefont [1]{#1}%
\providecommand \citenamefont [1]{#1}%
\providecommand \href@noop [0]{\@secondoftwo}%
\providecommand \href [0]{\begingroup \@sanitize@url \@href}%
\providecommand \@href[1]{\@@startlink{#1}\@@href}%
\providecommand \@@href[1]{\endgroup#1\@@endlink}%
\providecommand \@sanitize@url [0]{\catcode `\\12\catcode `\$12\catcode
  `\&12\catcode `\#12\catcode `\^12\catcode `\_12\catcode `\%12\relax}%
\providecommand \@@startlink[1]{}%
\providecommand \@@endlink[0]{}%
\providecommand \url  [0]{\begingroup\@sanitize@url \@url }%
\providecommand \@url [1]{\endgroup\@href {#1}{\urlprefix }}%
\providecommand \urlprefix  [0]{URL }%
\providecommand \Eprint [0]{\href }%
\providecommand \doibase [0]{http://dx.doi.org/}%
\providecommand \selectlanguage [0]{\@gobble}%
\providecommand \bibinfo  [0]{\@secondoftwo}%
\providecommand \bibfield  [0]{\@secondoftwo}%
\providecommand \translation [1]{[#1]}%
\providecommand \BibitemOpen [0]{}%
\providecommand \bibitemStop [0]{}%
\providecommand \bibitemNoStop [0]{.\EOS\space}%
\providecommand \EOS [0]{\spacefactor3000\relax}%
\providecommand \BibitemShut  [1]{\csname bibitem#1\endcsname}%
\let\auto@bib@innerbib\@empty
\bibitem [{\citenamefont {Kimble}(2008)}]{kimble2008quantum}%
  \BibitemOpen
  \bibfield  {author} {\bibinfo {author} {\bibfnamefont {H.~J.}\ \bibnamefont
  {Kimble}},\ }\bibfield  {title} {\enquote {\bibinfo {title} {The quantum
  internet},}\ }\href {https://www.nature.com/articles/nature07127} {\bibfield
  {journal} {\bibinfo  {journal} {Nature}\ }\textbf {\bibinfo {volume} {453}},\
  \bibinfo {pages} {1023--1030} (\bibinfo {year} {2008})}\BibitemShut {NoStop}%
\bibitem [{\citenamefont {Nickerson}\ \emph {et~al.}(2013)\citenamefont
  {Nickerson}, \citenamefont {Li},\ and\ \citenamefont
  {Benjamin}}]{nickerson2013topological}%
  \BibitemOpen
  \bibfield  {author} {\bibinfo {author} {\bibfnamefont {N.~H.}\ \bibnamefont
  {Nickerson}}, \bibinfo {author} {\bibfnamefont {Y.}~\bibnamefont {Li}}, \
  and\ \bibinfo {author} {\bibfnamefont {S.~C.}\ \bibnamefont {Benjamin}},\
  }\bibfield  {title} {\enquote {\bibinfo {title} {Topological quantum
  computing with a very noisy network and local error rates approaching one
  percent},}\ }\href {https://www.nature.com/articles/ncomms2773} {\bibfield
  {journal} {\bibinfo  {journal} {Nat. Commun.}\ }\textbf {\bibinfo {volume}
  {4}},\ \bibinfo {pages} {1756} (\bibinfo {year} {2013})}\BibitemShut
  {NoStop}%
\bibitem [{\citenamefont {Sasaki}\ \emph {et~al.}(2011)\citenamefont {Sasaki},
  \citenamefont {Fujiwara}, \citenamefont {Ishizuka}, \citenamefont {Klaus},
  \citenamefont {Wakui}, \citenamefont {Takeoka}, \citenamefont {Miki},
  \citenamefont {Yamashita}, \citenamefont {Wang}, \citenamefont {Tanaka} \emph
  {et~al.}}]{Sasaki2011field}%
  \BibitemOpen
  \bibfield  {author} {\bibinfo {author} {\bibfnamefont {M.}~\bibnamefont
  {Sasaki}}, \bibinfo {author} {\bibfnamefont {M.}~\bibnamefont {Fujiwara}},
  \bibinfo {author} {\bibfnamefont {H.}~\bibnamefont {Ishizuka}}, \bibinfo
  {author} {\bibfnamefont {W.}~\bibnamefont {Klaus}}, \bibinfo {author}
  {\bibfnamefont {K.}~\bibnamefont {Wakui}}, \bibinfo {author} {\bibfnamefont
  {M.}~\bibnamefont {Takeoka}}, \bibinfo {author} {\bibfnamefont
  {S.}~\bibnamefont {Miki}}, \bibinfo {author} {\bibfnamefont {T.}~\bibnamefont
  {Yamashita}}, \bibinfo {author} {\bibfnamefont {Z.}~\bibnamefont {Wang}},
  \bibinfo {author} {\bibfnamefont {A.}~\bibnamefont {Tanaka}},  \emph
  {et~al.},\ }\bibfield  {title} {\enquote {\bibinfo {title} {Field test of
  quantum key distribution in the {T}okyo {QKD} network},}\ }\href
  {https://www.osapublishing.org/oe/abstract.cfm?uri=oe-19-11-10387} {\bibfield
   {journal} {\bibinfo  {journal} {Opt. Express}\ }\textbf {\bibinfo {volume}
  {19}},\ \bibinfo {pages} {10387--10409} (\bibinfo {year} {2011})}\BibitemShut
  {NoStop}%
\bibitem [{\citenamefont {Wang}\ \emph {et~al.}(2013)\citenamefont {Wang},
  \citenamefont {Yang}, \citenamefont {Liao}, \citenamefont {Zhang},
  \citenamefont {Shen}, \citenamefont {Hu}, \citenamefont {Wu}, \citenamefont
  {Yang}, \citenamefont {Jiang}, \citenamefont {Tang} \emph
  {et~al.}}]{wang2013direct}%
  \BibitemOpen
  \bibfield  {author} {\bibinfo {author} {\bibfnamefont {J.-Y.}\ \bibnamefont
  {Wang}}, \bibinfo {author} {\bibfnamefont {B.}~\bibnamefont {Yang}}, \bibinfo
  {author} {\bibfnamefont {S.-K.}\ \bibnamefont {Liao}}, \bibinfo {author}
  {\bibfnamefont {L.}~\bibnamefont {Zhang}}, \bibinfo {author} {\bibfnamefont
  {Q.}~\bibnamefont {Shen}}, \bibinfo {author} {\bibfnamefont {X.-F.}\
  \bibnamefont {Hu}}, \bibinfo {author} {\bibfnamefont {J.-C.}\ \bibnamefont
  {Wu}}, \bibinfo {author} {\bibfnamefont {S.-J.}\ \bibnamefont {Yang}},
  \bibinfo {author} {\bibfnamefont {H.}~\bibnamefont {Jiang}}, \bibinfo
  {author} {\bibfnamefont {Y.-L.}\ \bibnamefont {Tang}},  \emph {et~al.},\
  }\bibfield  {title} {\enquote {\bibinfo {title} {Direct and full-scale
  experimental verifications towards ground-satellite quantum key
  distribution},}\ }\href {https://www.nature.com/articles/nphoton.2013.89}
  {\bibfield  {journal} {\bibinfo  {journal} {Nat. Photon.}\ }\textbf {\bibinfo
  {volume} {7}},\ \bibinfo {pages} {387--393} (\bibinfo {year}
  {2013})}\BibitemShut {NoStop}%
\bibitem [{\citenamefont {Steinert}\ \emph {et~al.}(2010)\citenamefont
  {Steinert}, \citenamefont {Dolde}, \citenamefont {Neumann}, \citenamefont
  {Aird}, \citenamefont {Naydenov}, \citenamefont {Balasubramanian},
  \citenamefont {Jelezko},\ and\ \citenamefont {Wrachtrup}}]{steinert2010high}%
  \BibitemOpen
  \bibfield  {author} {\bibinfo {author} {\bibfnamefont {S.}~\bibnamefont
  {Steinert}}, \bibinfo {author} {\bibfnamefont {F.}~\bibnamefont {Dolde}},
  \bibinfo {author} {\bibfnamefont {P.}~\bibnamefont {Neumann}}, \bibinfo
  {author} {\bibfnamefont {A.}~\bibnamefont {Aird}}, \bibinfo {author}
  {\bibfnamefont {B.}~\bibnamefont {Naydenov}}, \bibinfo {author}
  {\bibfnamefont {G.}~\bibnamefont {Balasubramanian}}, \bibinfo {author}
  {\bibfnamefont {F.}~\bibnamefont {Jelezko}}, \ and\ \bibinfo {author}
  {\bibfnamefont {J.}~\bibnamefont {Wrachtrup}},\ }\bibfield  {title} {\enquote
  {\bibinfo {title} {High sensitivity magnetic imaging using an array of spins
  in diamond},}\ }\href {http://aip.scitation.org/doi/10.1063/1.3385689}
  {\bibfield  {journal} {\bibinfo  {journal} {Rev. Sci. Instrum.}\ }\textbf
  {\bibinfo {volume} {81}},\ \bibinfo {pages} {043705} (\bibinfo {year}
  {2010})}\BibitemShut {NoStop}%
\bibitem [{\citenamefont {Hall}\ \emph {et~al.}(2012)\citenamefont {Hall},
  \citenamefont {Beart}, \citenamefont {Thomas}, \citenamefont {Simpson},
  \citenamefont {McGuinness}, \citenamefont {Cole}, \citenamefont {Manton},
  \citenamefont {Scholten}, \citenamefont {Jelezko}, \citenamefont {Wrachtrup}
  \emph {et~al.}}]{hall2012high}%
  \BibitemOpen
  \bibfield  {author} {\bibinfo {author} {\bibfnamefont {L.~T.}\ \bibnamefont
  {Hall}}, \bibinfo {author} {\bibfnamefont {G.~C.~G.}\ \bibnamefont {Beart}},
  \bibinfo {author} {\bibfnamefont {E.~A.}\ \bibnamefont {Thomas}}, \bibinfo
  {author} {\bibfnamefont {D.~A.}\ \bibnamefont {Simpson}}, \bibinfo {author}
  {\bibfnamefont {L.~P.}\ \bibnamefont {McGuinness}}, \bibinfo {author}
  {\bibfnamefont {J.~H.}\ \bibnamefont {Cole}}, \bibinfo {author}
  {\bibfnamefont {J.~H.}\ \bibnamefont {Manton}}, \bibinfo {author}
  {\bibfnamefont {R.~E.}\ \bibnamefont {Scholten}}, \bibinfo {author}
  {\bibfnamefont {F.}~\bibnamefont {Jelezko}}, \bibinfo {author} {\bibfnamefont
  {J.}~\bibnamefont {Wrachtrup}},  \emph {et~al.},\ }\bibfield  {title}
  {\enquote {\bibinfo {title} {High spatial and temporal resolution wide-field
  imaging of neuron activity using quantum nv-diamond},}\ }\href
  {https://www.nature.com/articles/srep00401} {\bibfield  {journal} {\bibinfo
  {journal} {Sci. Rep.}\ }\textbf {\bibinfo {volume} {2}} (\bibinfo {year}
  {2012})}\BibitemShut {NoStop}%
\bibitem [{\citenamefont {Pham}\ \emph {et~al.}(2011)\citenamefont {Pham},
  \citenamefont {Le~Sage}, \citenamefont {Stanwix}, \citenamefont {Yeung},
  \citenamefont {Glenn}, \citenamefont {Trifonov}, \citenamefont {Cappellaro},
  \citenamefont {Hemmer}, \citenamefont {Lukin}, \citenamefont {Park} \emph
  {et~al.}}]{pham2011magnetic}%
  \BibitemOpen
  \bibfield  {author} {\bibinfo {author} {\bibfnamefont {L.~M.}\ \bibnamefont
  {Pham}}, \bibinfo {author} {\bibfnamefont {D.}~\bibnamefont {Le~Sage}},
  \bibinfo {author} {\bibfnamefont {P.~L.}\ \bibnamefont {Stanwix}}, \bibinfo
  {author} {\bibfnamefont {T.~K.}\ \bibnamefont {Yeung}}, \bibinfo {author}
  {\bibfnamefont {D.}~\bibnamefont {Glenn}}, \bibinfo {author} {\bibfnamefont
  {A.}~\bibnamefont {Trifonov}}, \bibinfo {author} {\bibfnamefont
  {P.}~\bibnamefont {Cappellaro}}, \bibinfo {author} {\bibfnamefont {P.~R.}\
  \bibnamefont {Hemmer}}, \bibinfo {author} {\bibfnamefont {M.~D.}\
  \bibnamefont {Lukin}}, \bibinfo {author} {\bibfnamefont {H.}~\bibnamefont
  {Park}},  \emph {et~al.},\ }\bibfield  {title} {\enquote {\bibinfo {title}
  {Magnetic field imaging with nitrogen-vacancy ensembles},}\ }\href
  {http://iopscience.iop.org/article/10.1088/1367-2630/13/4/045021/meta}
  {\bibfield  {journal} {\bibinfo  {journal} {New J. Phys.}\ }\textbf {\bibinfo
  {volume} {13}},\ \bibinfo {pages} {045021} (\bibinfo {year}
  {2011})}\BibitemShut {NoStop}%
\bibitem [{\citenamefont {Seo}\ \emph {et~al.}(2007)\citenamefont {Seo},
  \citenamefont {Adam}, \citenamefont {Kang}, \citenamefont {Lee},
  \citenamefont {Jeoung}, \citenamefont {Park}, \citenamefont {Planken},\ and\
  \citenamefont {Kim}}]{seo2007fourier}%
  \BibitemOpen
  \bibfield  {author} {\bibinfo {author} {\bibfnamefont {M.~A.}\ \bibnamefont
  {Seo}}, \bibinfo {author} {\bibfnamefont {A.~J.~L.}\ \bibnamefont {Adam}},
  \bibinfo {author} {\bibfnamefont {J.~H.}\ \bibnamefont {Kang}}, \bibinfo
  {author} {\bibfnamefont {J.~W.}\ \bibnamefont {Lee}}, \bibinfo {author}
  {\bibfnamefont {S.~C.}\ \bibnamefont {Jeoung}}, \bibinfo {author}
  {\bibfnamefont {Q.~H.}\ \bibnamefont {Park}}, \bibinfo {author}
  {\bibfnamefont {P.~C.~M.}\ \bibnamefont {Planken}}, \ and\ \bibinfo {author}
  {\bibfnamefont {D.~S.}\ \bibnamefont {Kim}},\ }\bibfield  {title} {\enquote
  {\bibinfo {title} {Fourier-transform terahertz near-field imaging of
  one-dimensional slit arrays: mapping of electric-field-, magnetic-field-, and
  poynting vectors},}\ }\href
  {https://www.osapublishing.org/oe/abstract.cfm?uri=oe-15-19-11781} {\bibfield
   {journal} {\bibinfo  {journal} {Opt. Express}\ }\textbf {\bibinfo {volume}
  {15}},\ \bibinfo {pages} {11781--11789} (\bibinfo {year} {2007})}\BibitemShut
  {NoStop}%
\bibitem [{\citenamefont {Baumgratz}\ and\ \citenamefont
  {Datta}(2016)}]{baumgratz2016quantum}%
  \BibitemOpen
  \bibfield  {author} {\bibinfo {author} {\bibfnamefont {T.}~\bibnamefont
  {Baumgratz}}\ and\ \bibinfo {author} {\bibfnamefont {A.}~\bibnamefont
  {Datta}},\ }\bibfield  {title} {\enquote {\bibinfo {title} {Quantum enhanced
  estimation of a multidimensional field},}\ }\href
  {https://journals.aps.org/prl/abstract/10.1103/PhysRevLett.116.030801}
  {\bibfield  {journal} {\bibinfo  {journal} {Phys. Rev. Lett.}\ }\textbf
  {\bibinfo {volume} {116}},\ \bibinfo {pages} {030801} (\bibinfo {year}
  {2016})}\BibitemShut {NoStop}%
\bibitem [{\citenamefont {Humphreys}\ \emph {et~al.}(2013)\citenamefont
  {Humphreys}, \citenamefont {Barbieri}, \citenamefont {Datta},\ and\
  \citenamefont {Walmsley}}]{humphreys2013quantum}%
  \BibitemOpen
  \bibfield  {author} {\bibinfo {author} {\bibfnamefont {P.~C.}\ \bibnamefont
  {Humphreys}}, \bibinfo {author} {\bibfnamefont {M.}~\bibnamefont {Barbieri}},
  \bibinfo {author} {\bibfnamefont {A.}~\bibnamefont {Datta}}, \ and\ \bibinfo
  {author} {\bibfnamefont {I.~A.}\ \bibnamefont {Walmsley}},\ }\bibfield
  {title} {\enquote {\bibinfo {title} {Quantum enhanced multiple phase
  estimation},}\ }\href
  {https://journals.aps.org/prl/abstract/10.1103/PhysRevLett.111.070403}
  {\bibfield  {journal} {\bibinfo  {journal} {Phys. Rev. Lett.}\ }\textbf
  {\bibinfo {volume} {111}},\ \bibinfo {pages} {070403} (\bibinfo {year}
  {2013})}\BibitemShut {NoStop}%
\bibitem [{\citenamefont {Liu}\ \emph {et~al.}(2016)\citenamefont {Liu},
  \citenamefont {Lu}, \citenamefont {Sun},\ and\ \citenamefont
  {Wang}}]{liu2014quantum}%
  \BibitemOpen
  \bibfield  {author} {\bibinfo {author} {\bibfnamefont {J.}~\bibnamefont
  {Liu}}, \bibinfo {author} {\bibfnamefont {X.-M.}\ \bibnamefont {Lu}},
  \bibinfo {author} {\bibfnamefont {Z.}~\bibnamefont {Sun}}, \ and\ \bibinfo
  {author} {\bibfnamefont {X.}~\bibnamefont {Wang}},\ }\bibfield  {title}
  {\enquote {\bibinfo {title} {Quantum multiparameter metrology with
  generalized entangled coherent state},}\ }\href
  {http://iopscience.iop.org/article/10.1088/1751-8113/49/11/115302} {\bibfield
   {journal} {\bibinfo  {journal} {J. Phys. A: Math. Theor.}\ }\textbf
  {\bibinfo {volume} {49}},\ \bibinfo {pages} {115302} (\bibinfo {year}
  {2016})}\BibitemShut {NoStop}%
\bibitem [{\citenamefont {Yue}\ \emph {et~al.}(2014)\citenamefont {Yue},
  \citenamefont {Zhang},\ and\ \citenamefont {Fan}}]{yue2014quantum}%
  \BibitemOpen
  \bibfield  {author} {\bibinfo {author} {\bibfnamefont {J.-D.}\ \bibnamefont
  {Yue}}, \bibinfo {author} {\bibfnamefont {Y.-R.}\ \bibnamefont {Zhang}}, \
  and\ \bibinfo {author} {\bibfnamefont {H.}~\bibnamefont {Fan}},\ }\bibfield
  {title} {\enquote {\bibinfo {title} {Quantum-enhanced metrology for multiple
  phase estimation with noise},}\ }\href
  {https://www.nature.com/articles/srep05933} {\bibfield  {journal} {\bibinfo
  {journal} {Sci. Rep.}\ }\textbf {\bibinfo {volume} {4}} (\bibinfo {year}
  {2014})}\BibitemShut {NoStop}%
\bibitem [{\citenamefont {Ciampini}\ \emph {et~al.}(2016)\citenamefont
  {Ciampini}, \citenamefont {Spagnolo}, \citenamefont {Vitelli}, \citenamefont
  {Pezz{\`e}}, \citenamefont {Smerzi},\ and\ \citenamefont
  {Sciarrino}}]{ciampini2015quantum}%
  \BibitemOpen
  \bibfield  {author} {\bibinfo {author} {\bibfnamefont {M.~A.}\ \bibnamefont
  {Ciampini}}, \bibinfo {author} {\bibfnamefont {N.}~\bibnamefont {Spagnolo}},
  \bibinfo {author} {\bibfnamefont {C.}~\bibnamefont {Vitelli}}, \bibinfo
  {author} {\bibfnamefont {L.}~\bibnamefont {Pezz{\`e}}}, \bibinfo {author}
  {\bibfnamefont {A.}~\bibnamefont {Smerzi}}, \ and\ \bibinfo {author}
  {\bibfnamefont {F.}~\bibnamefont {Sciarrino}},\ }\bibfield  {title} {\enquote
  {\bibinfo {title} {Quantum-enhanced multiparameter estimation in multiarm
  interferometers},}\ }\href {https://www.nature.com/articles/srep28881}
  {\bibfield  {journal} {\bibinfo  {journal} {Sci. Rep.}\ }\textbf {\bibinfo
  {volume} {6}} (\bibinfo {year} {2016})}\BibitemShut {NoStop}%
\bibitem [{\citenamefont {Knott}\ \emph {et~al.}(2016)\citenamefont {Knott},
  \citenamefont {Proctor}, \citenamefont {Hayes}, \citenamefont {Ralph},
  \citenamefont {Kok},\ and\ \citenamefont {Dunningham}}]{knott2016local}%
  \BibitemOpen
  \bibfield  {author} {\bibinfo {author} {\bibfnamefont {P.~A.}\ \bibnamefont
  {Knott}}, \bibinfo {author} {\bibfnamefont {T.~J.}\ \bibnamefont {Proctor}},
  \bibinfo {author} {\bibfnamefont {A.~J.}\ \bibnamefont {Hayes}}, \bibinfo
  {author} {\bibfnamefont {J.~F.}\ \bibnamefont {Ralph}}, \bibinfo {author}
  {\bibfnamefont {P.}~\bibnamefont {Kok}}, \ and\ \bibinfo {author}
  {\bibfnamefont {J.~A.}\ \bibnamefont {Dunningham}},\ }\bibfield  {title}
  {\enquote {\bibinfo {title} {Local versus global strategies in
  multi-parameter estimation},}\ }\href
  {https://journals.aps.org/pra/abstract/10.1103/PhysRevA.94.062312} {\bibfield
   {journal} {\bibinfo  {journal} {Phys. Rev. A}\ }\textbf {\bibinfo {volume}
  {94}},\ \bibinfo {pages} {062312} (\bibinfo {year} {2016})}\BibitemShut
  {NoStop}%
\bibitem [{\citenamefont {Gagatsos}\ \emph {et~al.}(2016)\citenamefont
  {Gagatsos}, \citenamefont {Branford},\ and\ \citenamefont
  {Datta}}]{gagatsos2016gaussian}%
  \BibitemOpen
  \bibfield  {author} {\bibinfo {author} {\bibfnamefont {C.~N.}\ \bibnamefont
  {Gagatsos}}, \bibinfo {author} {\bibfnamefont {D.}~\bibnamefont {Branford}},
  \ and\ \bibinfo {author} {\bibfnamefont {A.}~\bibnamefont {Datta}},\
  }\bibfield  {title} {\enquote {\bibinfo {title} {Gaussian systems for
  quantum-enhanced multiple phase estimation},}\ }\href
  {https://journals.aps.org/pra/abstract/10.1103/PhysRevA.94.042342} {\bibfield
   {journal} {\bibinfo  {journal} {Phys. Rev. A}\ }\textbf {\bibinfo {volume}
  {94}},\ \bibinfo {pages} {042342} (\bibinfo {year} {2016})}\BibitemShut
  {NoStop}%
\bibitem [{\citenamefont {Zhang}\ and\ \citenamefont
  {Chan}(2017)}]{zhang2017quantum}%
  \BibitemOpen
  \bibfield  {author} {\bibinfo {author} {\bibfnamefont {L.}~\bibnamefont
  {Zhang}}\ and\ \bibinfo {author} {\bibfnamefont {K.~W.~C.}\ \bibnamefont
  {Chan}},\ }\bibfield  {title} {\enquote {\bibinfo {title} {Quantum
  multiparameter estimation with generalized balanced multimode noon-like
  states},}\ }\href
  {https://journals.aps.org/pra/abstract/10.1103/PhysRevA.95.032321} {\bibfield
   {journal} {\bibinfo  {journal} {Phys. Rev. A}\ }\textbf {\bibinfo {volume}
  {95}},\ \bibinfo {pages} {032321} (\bibinfo {year} {2017})}\BibitemShut
  {NoStop}%
\bibitem [{\citenamefont {Komar}\ \emph {et~al.}(2014)\citenamefont {Komar},
  \citenamefont {Kessler}, \citenamefont {Bishof}, \citenamefont {Jiang},
  \citenamefont {S{\o}rensen}, \citenamefont {Ye},\ and\ \citenamefont
  {Lukin}}]{komar2014quantum}%
  \BibitemOpen
  \bibfield  {author} {\bibinfo {author} {\bibfnamefont {P.}~\bibnamefont
  {Komar}}, \bibinfo {author} {\bibfnamefont {E.~M.}\ \bibnamefont {Kessler}},
  \bibinfo {author} {\bibfnamefont {M.}~\bibnamefont {Bishof}}, \bibinfo
  {author} {\bibfnamefont {L.}~\bibnamefont {Jiang}}, \bibinfo {author}
  {\bibfnamefont {A.~S.}\ \bibnamefont {S{\o}rensen}}, \bibinfo {author}
  {\bibfnamefont {J.}~\bibnamefont {Ye}}, \ and\ \bibinfo {author}
  {\bibfnamefont {M.~D.}\ \bibnamefont {Lukin}},\ }\bibfield  {title} {\enquote
  {\bibinfo {title} {A quantum network of clocks},}\ }\href
  {https://www.nature.com/articles/nphys3000} {\bibfield  {journal} {\bibinfo
  {journal} {Nat. Phys.}\ } (\bibinfo {year} {2014})}\BibitemShut {NoStop}%
\bibitem [{\citenamefont {Eldredge}\ \emph {et~al.}(2016)\citenamefont
  {Eldredge}, \citenamefont {Foss-Feig}, \citenamefont {Rolston},\ and\
  \citenamefont {Gorshkov}}]{eldredge2016optimal}%
  \BibitemOpen
  \bibfield  {author} {\bibinfo {author} {\bibfnamefont {Z.}~\bibnamefont
  {Eldredge}}, \bibinfo {author} {\bibfnamefont {M.}~\bibnamefont {Foss-Feig}},
  \bibinfo {author} {\bibfnamefont {S.~L.}\ \bibnamefont {Rolston}}, \ and\
  \bibinfo {author} {\bibfnamefont {A.~V.}\ \bibnamefont {Gorshkov}},\
  }\bibfield  {title} {\enquote {\bibinfo {title} {Optimal and secure
  measurement protocols for quantum sensor networks},}\ }\href
  {https://arxiv.org/abs/1607.04646} {\bibfield  {journal} {\bibinfo  {journal}
  {arXiv preprint arXiv:1607.04646}\ } (\bibinfo {year} {2016})}\BibitemShut
  {NoStop}%
\bibitem [{\citenamefont {Kok}\ \emph {et~al.}(2017)\citenamefont {Kok},
  \citenamefont {Dunningham},\ and\ \citenamefont {Ralph}}]{kok2017role}%
  \BibitemOpen
  \bibfield  {author} {\bibinfo {author} {\bibfnamefont {P.}~\bibnamefont
  {Kok}}, \bibinfo {author} {\bibfnamefont {J.}~\bibnamefont {Dunningham}}, \
  and\ \bibinfo {author} {\bibfnamefont {J.~F.}\ \bibnamefont {Ralph}},\
  }\bibfield  {title} {\enquote {\bibinfo {title} {Role of entanglement in
  calibrating optical quantum gyroscopes},}\ }\href
  {https://journals.aps.org/pra/abstract/10.1103/PhysRevA.95.012326} {\bibfield
   {journal} {\bibinfo  {journal} {Phys. Rev. A}\ }\textbf {\bibinfo {volume}
  {95}},\ \bibinfo {pages} {012326} (\bibinfo {year} {2017})}\BibitemShut
  {NoStop}%
\bibitem [{\citenamefont {Pyrkov}\ and\ \citenamefont
  {Byrnes}(2013)}]{pyrkov2013entanglement}%
  \BibitemOpen
  \bibfield  {author} {\bibinfo {author} {\bibfnamefont {A.~N.}\ \bibnamefont
  {Pyrkov}}\ and\ \bibinfo {author} {\bibfnamefont {T.}~\bibnamefont
  {Byrnes}},\ }\bibfield  {title} {\enquote {\bibinfo {title} {Entanglement
  generation in quantum networks of bose--einstein condensates},}\ }\href
  {http://iopscience.iop.org/article/10.1088/1367-2630/15/9/093019} {\bibfield
  {journal} {\bibinfo  {journal} {New J. Phys.}\ }\textbf {\bibinfo {volume}
  {15}},\ \bibinfo {pages} {093019} (\bibinfo {year} {2013})}\BibitemShut
  {NoStop}%
\bibitem [{\citenamefont {Wallquist}\ \emph {et~al.}(2009)\citenamefont
  {Wallquist}, \citenamefont {Hammerer}, \citenamefont {Rabl}, \citenamefont
  {Lukin},\ and\ \citenamefont {Zoller}}]{wallquist2009hybrid}%
  \BibitemOpen
  \bibfield  {author} {\bibinfo {author} {\bibfnamefont {M.}~\bibnamefont
  {Wallquist}}, \bibinfo {author} {\bibfnamefont {K.}~\bibnamefont {Hammerer}},
  \bibinfo {author} {\bibfnamefont {P.}~\bibnamefont {Rabl}}, \bibinfo {author}
  {\bibfnamefont {M.}~\bibnamefont {Lukin}}, \ and\ \bibinfo {author}
  {\bibfnamefont {P.}~\bibnamefont {Zoller}},\ }\bibfield  {title} {\enquote
  {\bibinfo {title} {Hybrid quantum devices and quantum engineering},}\ }\href
  {http://iopscience.iop.org/article/10.1088/0031-8949/2009/T137/014001}
  {\bibfield  {journal} {\bibinfo  {journal} {Phys. Scripta}\ }\textbf
  {\bibinfo {volume} {2009}},\ \bibinfo {pages} {014001} (\bibinfo {year}
  {2009})}\BibitemShut {NoStop}%
\bibitem [{\citenamefont {Szczykulska}\ \emph {et~al.}(2016)\citenamefont
  {Szczykulska}, \citenamefont {Baumgratz},\ and\ \citenamefont
  {Datta}}]{szczykulska2016multi}%
  \BibitemOpen
  \bibfield  {author} {\bibinfo {author} {\bibfnamefont {M.}~\bibnamefont
  {Szczykulska}}, \bibinfo {author} {\bibfnamefont {T.}~\bibnamefont
  {Baumgratz}}, \ and\ \bibinfo {author} {\bibfnamefont {A.}~\bibnamefont
  {Datta}},\ }\bibfield  {title} {\enquote {\bibinfo {title} {Multi-parameter
  quantum metrology},}\ }\href
  {http://www.tandfonline.com/doi/full/10.1080/23746149.2016.1230476}
  {\bibfield  {journal} {\bibinfo  {journal} {Adv. Phys. X}\ }\textbf {\bibinfo
  {volume} {1}},\ \bibinfo {pages} {621--639} (\bibinfo {year}
  {2016})}\BibitemShut {NoStop}%
\bibitem [{\citenamefont {Giovannetti}\ \emph {et~al.}(2011)\citenamefont
  {Giovannetti}, \citenamefont {Lloyd},\ and\ \citenamefont
  {Maccone}}]{giovannetti2011advances}%
  \BibitemOpen
  \bibfield  {author} {\bibinfo {author} {\bibfnamefont {V.}~\bibnamefont
  {Giovannetti}}, \bibinfo {author} {\bibfnamefont {S.}~\bibnamefont {Lloyd}},
  \ and\ \bibinfo {author} {\bibfnamefont {L.}~\bibnamefont {Maccone}},\
  }\bibfield  {title} {\enquote {\bibinfo {title} {Advances in quantum
  metrology},}\ }\href {https://www.nature.com/articles/nphoton.2011.35}
  {\bibfield  {journal} {\bibinfo  {journal} {Nat. Photon.}\ }\textbf {\bibinfo
  {volume} {5}},\ \bibinfo {pages} {222--229} (\bibinfo {year}
  {2011})}\BibitemShut {NoStop}%
\bibitem [{\citenamefont {Zwierz}\ \emph {et~al.}(2010)\citenamefont {Zwierz},
  \citenamefont {P{\'e}rez-Delgado},\ and\ \citenamefont
  {Kok}}]{zwierz2010general}%
  \BibitemOpen
  \bibfield  {author} {\bibinfo {author} {\bibfnamefont {M.}~\bibnamefont
  {Zwierz}}, \bibinfo {author} {\bibfnamefont {C.~A.}\ \bibnamefont
  {P{\'e}rez-Delgado}}, \ and\ \bibinfo {author} {\bibfnamefont
  {P.}~\bibnamefont {Kok}},\ }\bibfield  {title} {\enquote {\bibinfo {title}
  {General optimality of the heisenberg limit for quantum metrology},}\ }\href
  {https://journals.aps.org/prl/abstract/10.1103/PhysRevLett.105.180402}
  {\bibfield  {journal} {\bibinfo  {journal} {Phys. Rev. Lett.}\ }\textbf
  {\bibinfo {volume} {105}},\ \bibinfo {pages} {180402} (\bibinfo {year}
  {2010})}\BibitemShut {NoStop}%
\bibitem [{\citenamefont {Demkowicz-Dobrza{\'n}ski}\ \emph
  {et~al.}(2012)\citenamefont {Demkowicz-Dobrza{\'n}ski}, \citenamefont
  {Ko{\l}ody{\'n}ski},\ and\ \citenamefont
  {Gu{\c{t}}{\u{a}}}}]{demkowicz2012elusive}%
  \BibitemOpen
  \bibfield  {author} {\bibinfo {author} {\bibfnamefont {R.}~\bibnamefont
  {Demkowicz-Dobrza{\'n}ski}}, \bibinfo {author} {\bibfnamefont
  {J.}~\bibnamefont {Ko{\l}ody{\'n}ski}}, \ and\ \bibinfo {author}
  {\bibfnamefont {M.}~\bibnamefont {Gu{\c{t}}{\u{a}}}},\ }\bibfield  {title}
  {\enquote {\bibinfo {title} {The elusive {H}eisenberg limit in
  quantum-enhanced metrology},}\ }\href
  {https://www.nature.com/articles/ncomms2067} {\bibfield  {journal} {\bibinfo
  {journal} {Nat. Commun.}\ }\textbf {\bibinfo {volume} {3}},\ \bibinfo {pages}
  {1063} (\bibinfo {year} {2012})}\BibitemShut {NoStop}%
\bibitem [{\citenamefont {Fujiwara}\ and\ \citenamefont
  {Nagaoka}(1995)}]{fujiwara1995quantum}%
  \BibitemOpen
  \bibfield  {author} {\bibinfo {author} {\bibfnamefont {A.}~\bibnamefont
  {Fujiwara}}\ and\ \bibinfo {author} {\bibfnamefont {H.}~\bibnamefont
  {Nagaoka}},\ }\bibfield  {title} {\enquote {\bibinfo {title} {Quantum
  {F}isher metric and estimation for pure state models},}\ }\href
  {http://www.sciencedirect.com/science/article/pii/0375960195002699}
  {\bibfield  {journal} {\bibinfo  {journal} {Phys. Lett. A}\ }\textbf
  {\bibinfo {volume} {201}},\ \bibinfo {pages} {119--124} (\bibinfo {year}
  {1995})}\BibitemShut {NoStop}%
\bibitem [{\citenamefont {Matsumoto}(2002)}]{matsumoto2002new}%
  \BibitemOpen
  \bibfield  {author} {\bibinfo {author} {\bibfnamefont {K.}~\bibnamefont
  {Matsumoto}},\ }\bibfield  {title} {\enquote {\bibinfo {title} {A new
  approach to the cram{\'e}r-rao-type bound of the pure-state model},}\ }\href
  {http://iopscience.iop.org/article/10.1088/0305-4470/35/13/307} {\bibfield
  {journal} {\bibinfo  {journal} {J. Phys. A}\ }\textbf {\bibinfo {volume}
  {35}},\ \bibinfo {pages} {3111} (\bibinfo {year} {2002})}\BibitemShut
  {NoStop}%
\bibitem [{\citenamefont {Helstrom}(1976)}]{helstrom1976quantum}%
  \BibitemOpen
  \bibfield  {author} {\bibinfo {author} {\bibfnamefont {C.~W.}\ \bibnamefont
  {Helstrom}},\ }\href@noop {} {\emph {\bibinfo {title} {Quantum detection and
  estimation theory}}}\ (\bibinfo  {publisher} {Academic press},\ \bibinfo
  {year} {1976})\BibitemShut {NoStop}%
\bibitem [{\citenamefont {Paris}(2009)}]{paris2009quantum}%
  \BibitemOpen
  \bibfield  {author} {\bibinfo {author} {\bibfnamefont {M.~G.~A.}\
  \bibnamefont {Paris}},\ }\bibfield  {title} {\enquote {\bibinfo {title}
  {Quantum estimation for quantum technology},}\ }\href
  {http://www.worldscientific.com/doi/abs/10.1142/S0219749909004839} {\bibfield
   {journal} {\bibinfo  {journal} {Int. J. Quantum Inf.}\ }\textbf {\bibinfo
  {volume} {7}},\ \bibinfo {pages} {125--137} (\bibinfo {year}
  {2009})}\BibitemShut {NoStop}%
\bibitem [{\citenamefont {Helstrom}(1967)}]{helstrom1967minimum}%
  \BibitemOpen
  \bibfield  {author} {\bibinfo {author} {\bibfnamefont {C.~W.}\ \bibnamefont
  {Helstrom}},\ }\bibfield  {title} {\enquote {\bibinfo {title} {Minimum
  mean-squared error of estimates in quantum statistics},}\ }\href
  {https://www.sciencedirect.com/science/article/pii/0375960167903660}
  {\bibfield  {journal} {\bibinfo  {journal} {Phys. Lett. A}\ }\textbf
  {\bibinfo {volume} {25}},\ \bibinfo {pages} {101--102} (\bibinfo {year}
  {1967})}\BibitemShut {NoStop}%
\bibitem [{\citenamefont {Braunstein}\ and\ \citenamefont
  {Caves}(1994)}]{braunstein1994statistical}%
  \BibitemOpen
  \bibfield  {author} {\bibinfo {author} {\bibfnamefont {S.~L.}\ \bibnamefont
  {Braunstein}}\ and\ \bibinfo {author} {\bibfnamefont {C.~M.}\ \bibnamefont
  {Caves}},\ }\bibfield  {title} {\enquote {\bibinfo {title} {Statistical
  distance and the geometry of quantum states},}\ }\href
  {https://journals.aps.org/prl/abstract/10.1103/PhysRevLett.72.3439}
  {\bibfield  {journal} {\bibinfo  {journal} {Phys. Rev. Lett.}\ }\textbf
  {\bibinfo {volume} {72}},\ \bibinfo {pages} {3439--3443} (\bibinfo {year}
  {1994})}\BibitemShut {NoStop}%
\bibitem [{\citenamefont {Demkowicz-Dobrza{\'n}ski}\ \emph
  {et~al.}(2015)\citenamefont {Demkowicz-Dobrza{\'n}ski}, \citenamefont
  {Jarzyna},\ and\ \citenamefont {Ko{\l}ody{\'n}ski}}]{demkowicz2014quantum}%
  \BibitemOpen
  \bibfield  {author} {\bibinfo {author} {\bibfnamefont {R.}~\bibnamefont
  {Demkowicz-Dobrza{\'n}ski}}, \bibinfo {author} {\bibfnamefont
  {M.}~\bibnamefont {Jarzyna}}, \ and\ \bibinfo {author} {\bibfnamefont
  {J.}~\bibnamefont {Ko{\l}ody{\'n}ski}},\ }\bibfield  {title} {\enquote
  {\bibinfo {title} {Chapter four: Quantum limits in optical interferometry},}\
  }\href {https://www.sciencedirect.com/science/article/pii/S0079663815000049}
  {\bibfield  {journal} {\bibinfo  {journal} {Prog. Opt.}\ }\textbf {\bibinfo
  {volume} {60}},\ \bibinfo {pages} {345--435} (\bibinfo {year}
  {2015})}\BibitemShut {NoStop}%
\bibitem [{\citenamefont {Vidrighin}\ \emph {et~al.}(2014)\citenamefont
  {Vidrighin}, \citenamefont {Donati}, \citenamefont {Genoni}, \citenamefont
  {Jin}, \citenamefont {Kolthammer}, \citenamefont {Kim}, \citenamefont
  {Datta}, \citenamefont {Barbieri},\ and\ \citenamefont
  {Walmsley}}]{vidrighin2014joint}%
  \BibitemOpen
  \bibfield  {author} {\bibinfo {author} {\bibfnamefont {M.~D.}\ \bibnamefont
  {Vidrighin}}, \bibinfo {author} {\bibfnamefont {G.}~\bibnamefont {Donati}},
  \bibinfo {author} {\bibfnamefont {M.~G.}\ \bibnamefont {Genoni}}, \bibinfo
  {author} {\bibfnamefont {X.-M.}\ \bibnamefont {Jin}}, \bibinfo {author}
  {\bibfnamefont {W.~S.}\ \bibnamefont {Kolthammer}}, \bibinfo {author}
  {\bibfnamefont {M.~S.}\ \bibnamefont {Kim}}, \bibinfo {author} {\bibfnamefont
  {A.}~\bibnamefont {Datta}}, \bibinfo {author} {\bibfnamefont
  {M.}~\bibnamefont {Barbieri}}, \ and\ \bibinfo {author} {\bibfnamefont
  {I.~A.}\ \bibnamefont {Walmsley}},\ }\bibfield  {title} {\enquote {\bibinfo
  {title} {Joint estimation of phase and phase diffusion for quantum
  metrology},}\ }\href {https://www.nature.com/articles/ncomms4532} {\bibfield
  {journal} {\bibinfo  {journal} {Nat. Commun.}\ }\textbf {\bibinfo {volume}
  {5}} (\bibinfo {year} {2014})}\BibitemShut {NoStop}%
\bibitem [{\citenamefont {Fujiwara}(2001)}]{fujiwara2001estimation}%
  \BibitemOpen
  \bibfield  {author} {\bibinfo {author} {\bibfnamefont {A.}~\bibnamefont
  {Fujiwara}},\ }\bibfield  {title} {\enquote {\bibinfo {title} {Estimation of
  {SU(2)} operation and dense coding: An information geometric approach},}\
  }\href {https://journals.aps.org/pra/abstract/10.1103/PhysRevA.65.012316}
  {\bibfield  {journal} {\bibinfo  {journal} {Phys. Rev. A}\ }\textbf {\bibinfo
  {volume} {65}},\ \bibinfo {pages} {012316} (\bibinfo {year}
  {2001})}\BibitemShut {NoStop}%
\bibitem [{\citenamefont {Ragy}(2015)}]{ragy2015resources}%
  \BibitemOpen
  \bibfield  {author} {\bibinfo {author} {\bibfnamefont {S.}~\bibnamefont
  {Ragy}},\ }\emph {\bibinfo {title} {Resources in quantum imaging, detection
  and estimation}},\ \href@noop {} {Ph.D. thesis},\ \bibinfo  {school}
  {University of Nottingham} (\bibinfo {year} {2015})\BibitemShut {NoStop}%
\bibitem [{\citenamefont {Ragy}\ \emph {et~al.}(2016)\citenamefont {Ragy},
  \citenamefont {Jarzyna},\ and\ \citenamefont
  {Demkowicz-Dobrza{\'n}ski}}]{ragy2016compatibility}%
  \BibitemOpen
  \bibfield  {author} {\bibinfo {author} {\bibfnamefont {S.}~\bibnamefont
  {Ragy}}, \bibinfo {author} {\bibfnamefont {M.}~\bibnamefont {Jarzyna}}, \
  and\ \bibinfo {author} {\bibfnamefont {R.}~\bibnamefont
  {Demkowicz-Dobrza{\'n}ski}},\ }\bibfield  {title} {\enquote {\bibinfo {title}
  {Compatibility in multiparameter quantum metrology},}\ }\href
  {https://journals.aps.org/pra/abstract/10.1103/PhysRevA.94.052108} {\bibfield
   {journal} {\bibinfo  {journal} {Phys. Rev. A}\ }\textbf {\bibinfo {volume}
  {94}},\ \bibinfo {pages} {052108} (\bibinfo {year} {2016})}\BibitemShut
  {NoStop}%
\bibitem [{\citenamefont {Pezz{{\`e}}}\ \emph {et~al.}(2017)\citenamefont
  {Pezz{{\`e}}}, \citenamefont {Ciampini}, \citenamefont {Spagnolo},
  \citenamefont {Humphreys}, \citenamefont {Datta}, \citenamefont {Walmsley},
  \citenamefont {Barbieri}, \citenamefont {Sciarrino},\ and\ \citenamefont
  {Smerzi}}]{pezze2017optimal}%
  \BibitemOpen
  \bibfield  {author} {\bibinfo {author} {\bibfnamefont {L.}~\bibnamefont
  {Pezz{{\`e}}}}, \bibinfo {author} {\bibfnamefont {M.~A.}\ \bibnamefont
  {Ciampini}}, \bibinfo {author} {\bibfnamefont {N.}~\bibnamefont {Spagnolo}},
  \bibinfo {author} {\bibfnamefont {P.~C.}\ \bibnamefont {Humphreys}}, \bibinfo
  {author} {\bibfnamefont {A.}~\bibnamefont {Datta}}, \bibinfo {author}
  {\bibfnamefont {I.~A.}\ \bibnamefont {Walmsley}}, \bibinfo {author}
  {\bibfnamefont {M.}~\bibnamefont {Barbieri}}, \bibinfo {author}
  {\bibfnamefont {F.}~\bibnamefont {Sciarrino}}, \ and\ \bibinfo {author}
  {\bibfnamefont {A.}~\bibnamefont {Smerzi}},\ }\bibfield  {title} {\enquote
  {\bibinfo {title} {Optimal measurements for simultaneous quantum estimation
  of multiple phases},}\ }\href
  {https://journals.aps.org/prl/abstract/10.1103/PhysRevLett.119.130504}
  {\bibfield  {journal} {\bibinfo  {journal} {Phys. Rev. Lett.}\ }\textbf
  {\bibinfo {volume} {119}},\ \bibinfo {pages} {130504} (\bibinfo {year}
  {2017})}\BibitemShut {NoStop}%
\bibitem [{\citenamefont {Genoni}\ \emph {et~al.}(2013)\citenamefont {Genoni},
  \citenamefont {Paris}, \citenamefont {Adesso}, \citenamefont {Nha},
  \citenamefont {Knight},\ and\ \citenamefont {Kim}}]{genoni2013optimal}%
  \BibitemOpen
  \bibfield  {author} {\bibinfo {author} {\bibfnamefont {M.~G.}\ \bibnamefont
  {Genoni}}, \bibinfo {author} {\bibfnamefont {M.~G.~A.}\ \bibnamefont
  {Paris}}, \bibinfo {author} {\bibfnamefont {G.}~\bibnamefont {Adesso}},
  \bibinfo {author} {\bibfnamefont {H.}~\bibnamefont {Nha}}, \bibinfo {author}
  {\bibfnamefont {P.~L.}\ \bibnamefont {Knight}}, \ and\ \bibinfo {author}
  {\bibfnamefont {M.~S.}\ \bibnamefont {Kim}},\ }\bibfield  {title} {\enquote
  {\bibinfo {title} {Optimal estimation of joint parameters in phase space},}\
  }\href {https://journals.aps.org/pra/abstract/10.1103/PhysRevA.87.012107}
  {\bibfield  {journal} {\bibinfo  {journal} {Phys. Rev. A}\ }\textbf {\bibinfo
  {volume} {87}},\ \bibinfo {pages} {012107} (\bibinfo {year}
  {2013})}\BibitemShut {NoStop}%
\bibitem [{\citenamefont {Vaneph}\ \emph {et~al.}(2013)\citenamefont {Vaneph},
  \citenamefont {Tufarelli},\ and\ \citenamefont {Genoni}}]{vaneph2013quantum}%
  \BibitemOpen
  \bibfield  {author} {\bibinfo {author} {\bibfnamefont {C.}~\bibnamefont
  {Vaneph}}, \bibinfo {author} {\bibfnamefont {T.}~\bibnamefont {Tufarelli}}, \
  and\ \bibinfo {author} {\bibfnamefont {M.~G.}\ \bibnamefont {Genoni}},\
  }\bibfield  {title} {\enquote {\bibinfo {title} {Quantum estimation of a
  two-phase spin rotation},}\ }\href
  {https://www.degruyter.com/view/j/qmetro.2013.1.issue/qmetro-2013-0003/qmetro-2013-0003.xml}
  {\bibfield  {journal} {\bibinfo  {journal} {Quantum Measurements and Quantum
  Metrology}\ }\textbf {\bibinfo {volume} {1}},\ \bibinfo {pages} {12--20}
  (\bibinfo {year} {2013})}\BibitemShut {NoStop}%
\bibitem [{mpe()}]{mpe_prl_note1}%
  \BibitemOpen
  \href@noop {} {}\bibinfo {note} {Given a fixed time of evolution
  \cite{huelga1997improvement}.}\BibitemShut {Stop}%
\bibitem [{\citenamefont {Huelga}\ \emph {et~al.}(1997)\citenamefont {Huelga},
  \citenamefont {Macchiavello}, \citenamefont {Pellizzari}, \citenamefont
  {Plenio},\ and\ \citenamefont {Cirac}}]{huelga1997improvement}%
  \BibitemOpen
  \bibfield  {author} {\bibinfo {author} {\bibfnamefont {S.~F.}\ \bibnamefont
  {Huelga}}, \bibinfo {author} {\bibfnamefont {C.}~\bibnamefont
  {Macchiavello}}, \bibinfo {author} {\bibfnamefont {A.~K.}\ \bibnamefont
  {Pellizzari}, \bibfnamefont {T.and~Ekert}}, \bibinfo {author} {\bibfnamefont
  {M.~B.}\ \bibnamefont {Plenio}}, \ and\ \bibinfo {author} {\bibfnamefont
  {J.}~\bibnamefont {Cirac}},\ }\bibfield  {title} {\enquote {\bibinfo {title}
  {Improvement of frequency standards with quantum entanglement},}\ }\href
  {https://journals.aps.org/prl/abstract/10.1103/PhysRevLett.79.3865}
  {\bibfield  {journal} {\bibinfo  {journal} {Phys. Rev. Lett.}\ }\textbf
  {\bibinfo {volume} {79}},\ \bibinfo {pages} {3865} (\bibinfo {year}
  {1997})}\BibitemShut {NoStop}%
\bibitem [{\citenamefont {Tanaka}\ \emph {et~al.}(2015)\citenamefont {Tanaka},
  \citenamefont {Knott}, \citenamefont {Matsuzaki}, \citenamefont {Dooley},
  \citenamefont {Yamaguchi}, \citenamefont {Munro},\ and\ \citenamefont
  {Saito}}]{tanaka2014robust}%
  \BibitemOpen
  \bibfield  {author} {\bibinfo {author} {\bibfnamefont {T.}~\bibnamefont
  {Tanaka}}, \bibinfo {author} {\bibfnamefont {P.}~\bibnamefont {Knott}},
  \bibinfo {author} {\bibfnamefont {Y.}~\bibnamefont {Matsuzaki}}, \bibinfo
  {author} {\bibfnamefont {S.}~\bibnamefont {Dooley}}, \bibinfo {author}
  {\bibfnamefont {H.}~\bibnamefont {Yamaguchi}}, \bibinfo {author}
  {\bibfnamefont {W.~J.}\ \bibnamefont {Munro}}, \ and\ \bibinfo {author}
  {\bibfnamefont {S.}~\bibnamefont {Saito}},\ }\bibfield  {title} {\enquote
  {\bibinfo {title} {Proposed robust entanglement-based magnetic field sensor
  beyond the standard quantum limit},}\ }\href
  {https://journals.aps.org/prl/abstract/10.1103/PhysRevLett.115.170801}
  {\bibfield  {journal} {\bibinfo  {journal} {Phys. Rev. Lett.}\ }\textbf
  {\bibinfo {volume} {115}},\ \bibinfo {pages} {170801} (\bibinfo {year}
  {2015})}\BibitemShut {NoStop}%
\bibitem [{\citenamefont {Kessler}\ \emph {et~al.}(2014)\citenamefont
  {Kessler}, \citenamefont {Lovchinsky}, \citenamefont {Sushkov},\ and\
  \citenamefont {Lukin}}]{kessler2014quantum}%
  \BibitemOpen
  \bibfield  {author} {\bibinfo {author} {\bibfnamefont {E.~M.}\ \bibnamefont
  {Kessler}}, \bibinfo {author} {\bibfnamefont {I.}~\bibnamefont {Lovchinsky}},
  \bibinfo {author} {\bibfnamefont {A.~O.}\ \bibnamefont {Sushkov}}, \ and\
  \bibinfo {author} {\bibfnamefont {M.~D.}\ \bibnamefont {Lukin}},\ }\bibfield
  {title} {\enquote {\bibinfo {title} {Quantum error correction for
  metrology},}\ }\href
  {https://journals.aps.org/prl/abstract/10.1103/PhysRevLett.112.150802}
  {\bibfield  {journal} {\bibinfo  {journal} {Phys. Rev. Lett.}\ }\textbf
  {\bibinfo {volume} {112}},\ \bibinfo {pages} {150802} (\bibinfo {year}
  {2014})}\BibitemShut {NoStop}%
\bibitem [{\citenamefont {Liu}\ \emph {et~al.}(2015)\citenamefont {Liu},
  \citenamefont {Jing},\ and\ \citenamefont {Wang}}]{liu2015quantum}%
  \BibitemOpen
  \bibfield  {author} {\bibinfo {author} {\bibfnamefont {J.}~\bibnamefont
  {Liu}}, \bibinfo {author} {\bibfnamefont {X.-X.}\ \bibnamefont {Jing}}, \
  and\ \bibinfo {author} {\bibfnamefont {X.}~\bibnamefont {Wang}},\ }\bibfield
  {title} {\enquote {\bibinfo {title} {Quantum metrology with unitary
  parametrization processes},}\ }\href
  {https://www.nature.com/articles/srep08565} {\bibfield  {journal} {\bibinfo
  {journal} {Sci. Rep.}\ }\textbf {\bibinfo {volume} {5}} (\bibinfo {year}
  {2015})}\BibitemShut {NoStop}%
\bibitem [{\citenamefont {Liu}\ \emph {et~al.}(2014)\citenamefont {Liu},
  \citenamefont {Xiong},\ and\ \citenamefont {Song}}]{liu2014fidelity}%
  \BibitemOpen
  \bibfield  {author} {\bibinfo {author} {\bibfnamefont {J.}~\bibnamefont
  {Liu}}, \bibinfo {author} {\bibfnamefont {H.-N.}\ \bibnamefont {Xiong}}, \
  and\ \bibinfo {author} {\bibfnamefont {X.}~\bibnamefont {Song}, \bibfnamefont
  {F.and~Wang}},\ }\bibfield  {title} {\enquote {\bibinfo {title} {Fidelity
  susceptibility and quantum fisher information for density operators with
  arbitrary ranks},}\ }\href
  {https://www.sciencedirect.com/science/article/pii/S0378437114003926}
  {\bibfield  {journal} {\bibinfo  {journal} {Physica A}\ }\textbf {\bibinfo
  {volume} {410}},\ \bibinfo {pages} {167--173} (\bibinfo {year}
  {2014})}\BibitemShut {NoStop}%
\bibitem [{\citenamefont {Jarzyna}\ and\ \citenamefont
  {Demkowicz-Dobrza{\'n}ski}(2012)}]{jarzyna2012quantum}%
  \BibitemOpen
  \bibfield  {author} {\bibinfo {author} {\bibfnamefont {M.}~\bibnamefont
  {Jarzyna}}\ and\ \bibinfo {author} {\bibfnamefont {R.}~\bibnamefont
  {Demkowicz-Dobrza{\'n}ski}},\ }\bibfield  {title} {\enquote {\bibinfo {title}
  {Quantum interferometry with and without an external phase reference},}\
  }\href {https://journals.aps.org/pra/abstract/10.1103/PhysRevA.85.011801}
  {\bibfield  {journal} {\bibinfo  {journal} {Phys. Rev. A}\ }\textbf {\bibinfo
  {volume} {85}},\ \bibinfo {pages} {011801} (\bibinfo {year}
  {2012})}\BibitemShut {NoStop}%
\bibitem [{\citenamefont {Zhuang}\ \emph {et~al.}(2017)\citenamefont {Zhuang},
  \citenamefont {Zhang},\ and\ \citenamefont
  {Shapiro}}]{zhuang2017entanglement}%
  \BibitemOpen
  \bibfield  {author} {\bibinfo {author} {\bibfnamefont {Q.}~\bibnamefont
  {Zhuang}}, \bibinfo {author} {\bibfnamefont {Z.}~\bibnamefont {Zhang}}, \
  and\ \bibinfo {author} {\bibfnamefont {J.~H.}\ \bibnamefont {Shapiro}},\
  }\bibfield  {title} {\enquote {\bibinfo {title} {Entanglement-enhanced lidars
  for simultaneous range and velocity measurements},}\ }\href
  {https://journals.aps.org/pra/pdf/10.1103/PhysRevA.96.040304} {\bibfield
  {journal} {\bibinfo  {journal} {Phys. Rev. A}\ }\textbf {\bibinfo {volume}
  {96}} (\bibinfo {year} {2017})}\BibitemShut {NoStop}%
\bibitem [{\citenamefont {Ballester}(2004)}]{ballester2004estimation}%
  \BibitemOpen
  \bibfield  {author} {\bibinfo {author} {\bibfnamefont {M.~A.}\ \bibnamefont
  {Ballester}},\ }\bibfield  {title} {\enquote {\bibinfo {title} {Estimation of
  unitary quantum operations},}\ }\href
  {https://journals.aps.org/pra/pdf/10.1103/PhysRevA.69.022303} {\bibfield
  {journal} {\bibinfo  {journal} {Phys. Rev. A}\ }\textbf {\bibinfo {volume}
  {69}},\ \bibinfo {pages} {022303} (\bibinfo {year} {2004})}\BibitemShut
  {NoStop}%
\bibitem [{\citenamefont {Aasi}\ \emph {et~al.}(2013)\citenamefont {Aasi},
  \citenamefont {Abadie}, \citenamefont {Abbott}, \citenamefont {Abbott},
  \citenamefont {Abbott}, \citenamefont {Abernathy}, \citenamefont {Adams},
  \citenamefont {Adams}, \citenamefont {Addesso}, \citenamefont {Adhikari}
  \emph {et~al.}}]{aasi2013enhanced}%
  \BibitemOpen
  \bibfield  {author} {\bibinfo {author} {\bibfnamefont {J.}~\bibnamefont
  {Aasi}}, \bibinfo {author} {\bibfnamefont {J.}~\bibnamefont {Abadie}},
  \bibinfo {author} {\bibfnamefont {B.~P.}\ \bibnamefont {Abbott}}, \bibinfo
  {author} {\bibfnamefont {R.}~\bibnamefont {Abbott}}, \bibinfo {author}
  {\bibfnamefont {T.~D.}\ \bibnamefont {Abbott}}, \bibinfo {author}
  {\bibfnamefont {MR}~\bibnamefont {Abernathy}}, \bibinfo {author}
  {\bibfnamefont {C.}~\bibnamefont {Adams}}, \bibinfo {author} {\bibfnamefont
  {T.}~\bibnamefont {Adams}}, \bibinfo {author} {\bibfnamefont
  {P.}~\bibnamefont {Addesso}}, \bibinfo {author} {\bibfnamefont {R.~X.}\
  \bibnamefont {Adhikari}},  \emph {et~al.},\ }\bibfield  {title} {\enquote
  {\bibinfo {title} {Enhanced sensitivity of the {LIGO} gravitational wave
  detector by using squeezed states of light},}\ }\href
  {https://www.nature.com/articles/nphoton.2013.177} {\bibfield  {journal}
  {\bibinfo  {journal} {Nat. Photon.}\ }\textbf {\bibinfo {volume} {7}},\
  \bibinfo {pages} {613--619} (\bibinfo {year} {2013})}\BibitemShut {NoStop}%
\bibitem [{\citenamefont {Zhang}\ \emph {et~al.}(2014)\citenamefont {Zhang},
  \citenamefont {Wang}, \citenamefont {Jing}, \citenamefont {Mu},\ and\
  \citenamefont {Fan}}]{zhang2014fitting}%
  \BibitemOpen
  \bibfield  {author} {\bibinfo {author} {\bibfnamefont {Yong-Liang}\
  \bibnamefont {Zhang}}, \bibinfo {author} {\bibfnamefont {Huan}\ \bibnamefont
  {Wang}}, \bibinfo {author} {\bibfnamefont {Li}~\bibnamefont {Jing}}, \bibinfo
  {author} {\bibfnamefont {Liang-Zhu}\ \bibnamefont {Mu}}, \ and\ \bibinfo
  {author} {\bibfnamefont {Heng}\ \bibnamefont {Fan}},\ }\bibfield  {title}
  {\enquote {\bibinfo {title} {Fitting magnetic field gradient with
  heisenberg-scaling accuracy},}\ }\href
  {https://www.nature.com/articles/srep07390} {\bibfield  {journal} {\bibinfo
  {journal} {Sci. Rep.}\ }\textbf {\bibinfo {volume} {4}} (\bibinfo {year}
  {2014})}\BibitemShut {NoStop}%
\bibitem [{\citenamefont {Ng}\ and\ \citenamefont {Kim}(2014)}]{ng2014quantum}%
  \BibitemOpen
  \bibfield  {author} {\bibinfo {author} {\bibfnamefont {H.~T.}\ \bibnamefont
  {Ng}}\ and\ \bibinfo {author} {\bibfnamefont {K.}~\bibnamefont {Kim}},\
  }\bibfield  {title} {\enquote {\bibinfo {title} {Quantum estimation of
  magnetic-field gradient using {W}-state},}\ }\href
  {https://www.sciencedirect.com/science/article/pii/S0030401814005951}
  {\bibfield  {journal} {\bibinfo  {journal} {Opt. Commun.}\ }\textbf {\bibinfo
  {volume} {331}},\ \bibinfo {pages} {353--358} (\bibinfo {year}
  {2014})}\BibitemShut {NoStop}%
\bibitem [{\citenamefont {Bollinger}\ \emph {et~al.}(1996)\citenamefont
  {Bollinger}, \citenamefont {Itano}, \citenamefont {Wineland},\ and\
  \citenamefont {Heinzen}}]{bollinger1996optimal}%
  \BibitemOpen
  \bibfield  {author} {\bibinfo {author} {\bibfnamefont {J.~J.}\ \bibnamefont
  {Bollinger}}, \bibinfo {author} {\bibfnamefont {W.~M.}\ \bibnamefont
  {Itano}}, \bibinfo {author} {\bibfnamefont {D.~J.}\ \bibnamefont {Wineland}},
  \ and\ \bibinfo {author} {\bibfnamefont {D.~J.}\ \bibnamefont {Heinzen}},\
  }\bibfield  {title} {\enquote {\bibinfo {title} {Optimal frequency
  measurements with maximally correlated states},}\ }\href
  {https://www.sciencedirect.com/science/article/pii/S0030401814005951}
  {\bibfield  {journal} {\bibinfo  {journal} {Phys. Rev. A}\ }\textbf {\bibinfo
  {volume} {54}},\ \bibinfo {pages} {R4649} (\bibinfo {year}
  {1996})}\BibitemShut {NoStop}%
\bibitem [{\citenamefont {Ge}\ \emph {et~al.}(2017)\citenamefont {Ge},
  \citenamefont {Jacobs}, \citenamefont {Eldredge}, \citenamefont {Gorshkov},\
  and\ \citenamefont {Foss-Feig}}]{ge2017distributed}%
  \BibitemOpen
  \bibfield  {author} {\bibinfo {author} {\bibfnamefont {W.}~\bibnamefont
  {Ge}}, \bibinfo {author} {\bibfnamefont {K.}~\bibnamefont {Jacobs}}, \bibinfo
  {author} {\bibfnamefont {Z.}~\bibnamefont {Eldredge}}, \bibinfo {author}
  {\bibfnamefont {A.~V.}\ \bibnamefont {Gorshkov}}, \ and\ \bibinfo {author}
  {\bibfnamefont {M.}~\bibnamefont {Foss-Feig}},\ }\bibfield  {title} {\enquote
  {\bibinfo {title} {Distributed quantum metrology and the entangling power of
  linear networks},}\ }\href {https://arxiv.org/abs/1707.06655} {\bibfield
  {journal} {\bibinfo  {journal} {arXiv preprint arXiv:1707.06655}\ } (\bibinfo
  {year} {2017})}\BibitemShut {NoStop}%
\bibitem [{\citenamefont {Nielsen}\ and\ \citenamefont
  {Chuang}(2010)}]{nielsen2010quantum}%
  \BibitemOpen
  \bibfield  {author} {\bibinfo {author} {\bibfnamefont {M.~A.}\ \bibnamefont
  {Nielsen}}\ and\ \bibinfo {author} {\bibfnamefont {I.~L.}\ \bibnamefont
  {Chuang}},\ }\href@noop {} {\emph {\bibinfo {title} {Quantum computation and
  quantum information}}}\ (\bibinfo  {publisher} {Cambridge University Press},\
  \bibinfo {year} {2010})\BibitemShut {NoStop}%
\bibitem [{\citenamefont {Wolfgramm}\ \emph {et~al.}(2013)\citenamefont
  {Wolfgramm}, \citenamefont {Vitelli}, \citenamefont {Beduini}, \citenamefont
  {Godbout},\ and\ \citenamefont {Mitchell}}]{wolfgramm2013entanglement}%
  \BibitemOpen
  \bibfield  {author} {\bibinfo {author} {\bibfnamefont {F.}~\bibnamefont
  {Wolfgramm}}, \bibinfo {author} {\bibfnamefont {C.}~\bibnamefont {Vitelli}},
  \bibinfo {author} {\bibfnamefont {F.~A.}\ \bibnamefont {Beduini}}, \bibinfo
  {author} {\bibfnamefont {N.}~\bibnamefont {Godbout}}, \ and\ \bibinfo
  {author} {\bibfnamefont {M.~W.}\ \bibnamefont {Mitchell}},\ }\bibfield
  {title} {\enquote {\bibinfo {title} {Entanglement-enhanced probing of a
  delicate material system},}\ }\href
  {https://www.nature.com/articles/nphoton.2012.300} {\bibfield  {journal}
  {\bibinfo  {journal} {Nat. Photon.}\ }\textbf {\bibinfo {volume} {7}},\
  \bibinfo {pages} {28--32} (\bibinfo {year} {2013})}\BibitemShut {NoStop}%
\bibitem [{\citenamefont {Carlton}\ \emph {et~al.}(2010)\citenamefont
  {Carlton}, \citenamefont {Boulanger}, \citenamefont {Kervrann}, \citenamefont
  {Sibarita}, \citenamefont {Salamero}, \citenamefont {Gordon-Messer},
  \citenamefont {Bressan}, \citenamefont {Haber}, \citenamefont {Haase},
  \citenamefont {Shao} \emph {et~al.}}]{carlton2010fast}%
  \BibitemOpen
  \bibfield  {author} {\bibinfo {author} {\bibfnamefont {P.~M.}\ \bibnamefont
  {Carlton}}, \bibinfo {author} {\bibfnamefont {J.}~\bibnamefont {Boulanger}},
  \bibinfo {author} {\bibfnamefont {C.}~\bibnamefont {Kervrann}}, \bibinfo
  {author} {\bibfnamefont {J.-B.}\ \bibnamefont {Sibarita}}, \bibinfo {author}
  {\bibfnamefont {J.}~\bibnamefont {Salamero}}, \bibinfo {author}
  {\bibfnamefont {S.}~\bibnamefont {Gordon-Messer}}, \bibinfo {author}
  {\bibfnamefont {D.}~\bibnamefont {Bressan}}, \bibinfo {author} {\bibfnamefont
  {J.~E.}\ \bibnamefont {Haber}}, \bibinfo {author} {\bibfnamefont
  {S.}~\bibnamefont {Haase}}, \bibinfo {author} {\bibfnamefont
  {L.}~\bibnamefont {Shao}},  \emph {et~al.},\ }\bibfield  {title} {\enquote
  {\bibinfo {title} {Fast live simultaneous multiwavelength four-dimensional
  optical microscopy},}\ }\href
  {http://www.pnas.org/content/107/37/16016.abstract} {\bibfield  {journal}
  {\bibinfo  {journal} {Proc. Natl. Acad. Sci.}\ }\textbf {\bibinfo {volume}
  {107}},\ \bibinfo {pages} {16016--16022} (\bibinfo {year}
  {2010})}\BibitemShut {NoStop}%
\bibitem [{\citenamefont {Taylor}\ \emph {et~al.}(2013)\citenamefont {Taylor},
  \citenamefont {Janousek}, \citenamefont {Daria}, \citenamefont {Knittel},
  \citenamefont {Hage}, \citenamefont {Bachor},\ and\ \citenamefont
  {Bowen}}]{taylor2013biological}%
  \BibitemOpen
  \bibfield  {author} {\bibinfo {author} {\bibfnamefont {M.~A.}\ \bibnamefont
  {Taylor}}, \bibinfo {author} {\bibfnamefont {J.}~\bibnamefont {Janousek}},
  \bibinfo {author} {\bibfnamefont {V.}~\bibnamefont {Daria}}, \bibinfo
  {author} {\bibfnamefont {J.}~\bibnamefont {Knittel}}, \bibinfo {author}
  {\bibfnamefont {B.}~\bibnamefont {Hage}}, \bibinfo {author} {\bibfnamefont
  {H.-A.}\ \bibnamefont {Bachor}}, \ and\ \bibinfo {author} {\bibfnamefont
  {W.~P.}\ \bibnamefont {Bowen}},\ }\bibfield  {title} {\enquote {\bibinfo
  {title} {Biological measurement beyond the quantum limit},}\ }\href
  {https://www.nature.com/articles/nphoton.2012.346} {\bibfield  {journal}
  {\bibinfo  {journal} {Nat. Photon.}\ }\textbf {\bibinfo {volume} {7}},\
  \bibinfo {pages} {229--233} (\bibinfo {year} {2013})}\BibitemShut {NoStop}%
\bibitem [{\citenamefont {Tey}\ \emph {et~al.}(2008)\citenamefont {Tey},
  \citenamefont {Chen}, \citenamefont {Aljunid}, \citenamefont {Chng},
  \citenamefont {Huber}, \citenamefont {Maslennikov},\ and\ \citenamefont
  {Kurtsiefer}}]{tey2008strong}%
  \BibitemOpen
  \bibfield  {author} {\bibinfo {author} {\bibfnamefont {M.~K.}\ \bibnamefont
  {Tey}}, \bibinfo {author} {\bibfnamefont {Z.}~\bibnamefont {Chen}}, \bibinfo
  {author} {\bibfnamefont {S.~A.}\ \bibnamefont {Aljunid}}, \bibinfo {author}
  {\bibfnamefont {B.}~\bibnamefont {Chng}}, \bibinfo {author} {\bibfnamefont
  {F.}~\bibnamefont {Huber}}, \bibinfo {author} {\bibfnamefont
  {G.}~\bibnamefont {Maslennikov}}, \ and\ \bibinfo {author} {\bibfnamefont
  {C.}~\bibnamefont {Kurtsiefer}},\ }\bibfield  {title} {\enquote {\bibinfo
  {title} {Strong interaction between light and a single trapped atom without
  the need for a cavity},}\ }\href {https://www.nature.com/articles/nphys1096}
  {\bibfield  {journal} {\bibinfo  {journal} {Nature Phys.}\ }\textbf {\bibinfo
  {volume} {4}},\ \bibinfo {pages} {924--927} (\bibinfo {year}
  {2008})}\BibitemShut {NoStop}%
\bibitem [{\citenamefont {Eckert}\ \emph {et~al.}(2008)\citenamefont {Eckert},
  \citenamefont {Romero-Isart}, \citenamefont {Rodriguez}, \citenamefont
  {Lewenstein},\ and\ \citenamefont {Sanpera}}]{eckert2008quantum}%
  \BibitemOpen
  \bibfield  {author} {\bibinfo {author} {\bibfnamefont {K.}~\bibnamefont
  {Eckert}}, \bibinfo {author} {\bibfnamefont {O.}~\bibnamefont
  {Romero-Isart}}, \bibinfo {author} {\bibfnamefont {M.}~\bibnamefont
  {Rodriguez}}, \bibinfo {author} {\bibfnamefont {E.~S.}\ \bibnamefont
  {Lewenstein}, \bibfnamefont {M.and~Polzik}}, \ and\ \bibinfo {author}
  {\bibfnamefont {A.}~\bibnamefont {Sanpera}},\ }\bibfield  {title} {\enquote
  {\bibinfo {title} {Quantum non-demolition detection of strongly correlated
  systems},}\ }\href {https://www.nature.com/articles/nphys776} {\bibfield
  {journal} {\bibinfo  {journal} {Nature Phys.}\ }\textbf {\bibinfo {volume}
  {4}},\ \bibinfo {pages} {50--54} (\bibinfo {year} {2008})}\BibitemShut
  {NoStop}%
\bibitem [{\citenamefont {Pototschnig}\ \emph {et~al.}(2011)\citenamefont
  {Pototschnig}, \citenamefont {Chassagneux}, \citenamefont {Hwang},
  \citenamefont {Zumofen}, \citenamefont {Renn},\ and\ \citenamefont
  {Sandoghdar}}]{pototschnig2011controlling}%
  \BibitemOpen
  \bibfield  {author} {\bibinfo {author} {\bibfnamefont {M.}~\bibnamefont
  {Pototschnig}}, \bibinfo {author} {\bibfnamefont {Y.}~\bibnamefont
  {Chassagneux}}, \bibinfo {author} {\bibfnamefont {J.}~\bibnamefont {Hwang}},
  \bibinfo {author} {\bibfnamefont {G.}~\bibnamefont {Zumofen}}, \bibinfo
  {author} {\bibfnamefont {A.}~\bibnamefont {Renn}}, \ and\ \bibinfo {author}
  {\bibfnamefont {V.}~\bibnamefont {Sandoghdar}},\ }\bibfield  {title}
  {\enquote {\bibinfo {title} {Controlling the phase of a light beam with a
  single molecule},}\ }\href
  {https://journals.aps.org/prl/abstract/10.1103/PhysRevLett.107.063001}
  {\bibfield  {journal} {\bibinfo  {journal} {Phys. Rev. Lett.}\ }\textbf
  {\bibinfo {volume} {107}},\ \bibinfo {pages} {063001} (\bibinfo {year}
  {2011})}\BibitemShut {NoStop}%
\bibitem [{prl()}]{prl_sm_fn}%
  \BibitemOpen
  \href@noop {} {}\bibinfo {note} {As $C>0$ then by definition
  $\boldsymbol{v}^TC\boldsymbol{v} > 0$ for any vector $\boldsymbol{v}$ except
  when $\boldsymbol{v}=0$. We have that
  $\boldsymbol{v}^TBCB^T\boldsymbol{v}=\boldsymbol{w}^TC\boldsymbol{w} >0$
  unless $\boldsymbol{w}=0$, where $\boldsymbol{w}=B^T\boldsymbol{v}$. Hence
  $BCB^T\geq 0$. For general $B$, $BCB^T$ is not guaranteed to be positive
  definite as we can have $\boldsymbol{w}=0$ for $\boldsymbol{v}\neq 0$.
  However, if $B$ is a invertible matrix then $\boldsymbol{w}=0$ only if
  $\boldsymbol{v}= 0$ and so $BCB^T> 0$.}\BibitemShut {Stop}%
\end{thebibliography}%

\appendix

\section{Proof of theorem 1}
In this section we complete the proof of Theorem 1 from the main text. The elements of the proof that we deferred to this appendix are: (A) A proof that any invertible QFIM $\mathcal{F}$ satisfies 
\begin{equation}
[\mathcal{F}^{-1}]_{[kk]} \geq [\mathcal{F}_{[kk]}]^{-1},
\label{eq:qfim-in-}
\end{equation}
where the meaning of this notation (already introduced in the main text) will be clarified below; (B) A proof that, for any QSN with commuting parameter generators, when the initial state of the QSN is the $\varphi$ state, introduced in the main text, the QCRB is saturated by a POVM that can be implemented by independent measurements at each of the sensors; and (C) A proof that for every mixed state of the QSN there exists a pure state of the QSN that, when used as the initial state for the estimation, results in an equal or lower estimation uncertainty $E_{\bs{\Phi}}$ (when the optimal measurement is used) whilst using up the same amount of resources. We now prove (A -- C) in turn, in Propositions 1 -- 3, respectively.

In order to make it clear exactly which sub-matrices of the QFIM and inverse QFIM we are referring to in Eq.~\eqref{eq:qfim-in-}, we begin by first giving a detailed explanation of our sub-vector and sub-matrix notation, that we introduced briefly in the main text.  As in the main text, consider ``partitioning'' the $d$-dimensional vector $\bs{\phi}$ into $m$ sub-vectors, where the $k$\textsuperscript{th} sub-vector has a dimension of $d_k$ and $d=d_1+\dots+d_m$. More specifically, let the 1\textsuperscript{st} sub-vector, denoted $\bs{\phi}_{[1]}$, be given by 
\begin{equation}
\bs{\phi}_{[1]} :=(\phi_{1},\dots,\phi_{d_1})^T,
\end{equation}
 let the 2\textsuperscript{nd} sub-vector be 
 \begin{equation}
 \bs{\phi}_{[2]} :=(\phi_{1+d_1},\dots,\phi_{d_1+d_2})^T,
 \end{equation}
  and so on. Therefore, by denoting $d_{<k} := d_1 + d_2 + \dots + d_{k-1}$,
the $k$\textsuperscript{th} sub-vector is given by
\begin{equation}
 \bs{\phi}_{[k]}:=(\phi_{(1+d_{<k})},\dots,\phi_{(d_k + d_{<k})})^T.
  \end{equation}

Using an analogous notation, for a $d\times d$ matrix $M$ and a given partitioning of $d$ into $d=d_1+\dots+d_m$, we let $M_{[jk]}$ denote the sub-matrix of $M$ obtained by removing the elements that are not both in rows $1+d_{<j}$ to $d_j+d_{<j}$ and columns $1+d_{<k}$ to $d_k+d_{<k}$. Hence, 
\begin{equation}
M =  \begin{pmatrix}
  M_{[11]} & M_{[12]} & \cdots & M_{[1m]} \\
  M_{[21]} & M_{[22]} & \cdots & M_{[2m]} \\
  \vdots  & \vdots  & \ddots & \vdots  \\
  M_{[m1]} & M_{[m2]} & \cdots & M_{[mm]}
 \end{pmatrix}.
\end{equation}
Note that the parentheses in the subscripts of this notation are used to denote that these are sub-vectors and sub-matrices of $\bs{\phi}$ and $M$, respectively, and not just the ordinary scalar elements of $\bs{\phi}$ and $M$. It will be useful to define 
\begin{equation}
\mathbb{P}_j := \{ 1+d_{<j}, 2+d_{<j}, \dots, d_j+d_{<j} \},
\label{eq:pj}
\end{equation}
 i.e, $\mathbb{P}_j$ contains the labels for the parameters in the $j$\textsuperscript{th} partition.

\begin{proposition} For any invertible QFIM $\mathcal{F}$ for a $d$-dimensional vector $\bs{\phi}$, and an arbitrary partitioning of this vector into sub-vectors, $\bs{\phi}_{[1]}$, $\bs{\phi}_{[2]},$ $\dots,$ $\bs{\phi}_{[m]}$,
\begin{equation}
 [\mathcal{F}^{-1}]_{[kk]} \geq \left[\mathcal{F}_{[kk]}\right]^{-1}, 
 \label{eq:matrix-in}
\end{equation}
for all $k=1,2,\dots,m$. Moreover, the equality is obtained for any particular $k$ if and only if $\mathcal{F}_{[jk]}=\mathcal{F}_{[kj]} =0$ for all $j\neq k$. 
\end{proposition}
To understand this statement and the following proof, it is important to note that for two matrices $A$ and $B$, $A\geq B$ and $A\neq B$ does \emph{not} imply that $A>B$.

\begin{proof} Any QFIM $\mathcal{F}$ is real, symmetric and positive semi-definite \cite{paris2009quantum}, and if it is invertible it is positive definite. As such, we can instead prove this proposition for an arbitrary finite-dimensional $d \times d$ real, symmetric and positive definite matrix, $A$, and an arbitrary partitioning $d=d_1+\dots + d_m$ of this matrix. For any such partitioning consider the $d\times d$ matrix $P_k$ defined by the action on an arbitrary vector $\bs v$:
 \begin{equation}
P_k \begin{pmatrix}
  \bs{v}_{[1]}  \\
  \vdots   \\
  \bs{v}_{[m-2]} \\
    \bs{v}_{[m-1]} \\
      \bs{v}_{[m]}
 \end{pmatrix} =\begin{pmatrix}
  \bs{v}_{[1]}  \\
  \vdots   \\
  \bs{v}_{[m-1]} \\
    \bs{v}_{[m]} \\
      \bs{v}_{[k]}
 \end{pmatrix}.
\end{equation}
 $P_k$ is a permutation matrix and hence $P_kP^T_k=\mathds{1}$. Consider the matrix $\tilde{A}(k)=P_{k} A P_k^T$. This $\tilde{A}(k)$ matrix is symmetric as $A$ is symmetric. For any $s \times s$ matrix, $C$, and $t \times s$ matrix, $B$, then
  \begin{equation} C>0 \implies BCB^T \geq 0, \label{cdsd} \end{equation}
 and if $B$ is a (square) invertible matrix then $BCB^T > 0$ \cite{prl_sm_fn}. Hence $\tilde{A}(k)>0$ because $A>0$ and $P_k$ is invertible. It may be confirmed that
 \begin{equation} \tilde{A}(k)= \begin{pmatrix} A_{[\neq k]} & A_{k}^T \\ A_{k} & A_{[kk]} \end{pmatrix},
 \end{equation}
 where $A_{[\neq k]}$ is a positive definite matrix consisting of those $A_{[mn]}$ matrices with $m \neq k $ and $n \neq k$ (its exact form is irrelevant) and $A_k=(A_{[k1]},A_{[k2]},\dots,A_{[km]})$ where the second label in the subscripts here takes each value sequentially except that it misses out $k$.

 Consider any matrix $M$ that is symmetric, positive definite and has the form
\begin{equation} M= \begin{pmatrix} a & b^T \\ b & c \end{pmatrix},\label{matrix-22}
\end{equation}
 where $a$ and $c$ are square matrices of any sizes and $b$ is of the appropriate dimensions to make this a valid matrix. $M>0$ implies that $a>0$ and $c>0$. The inverse of $M$ exists and is given explicitly by
\begin{equation} M^{-1}= \begin{pmatrix} a^{-1}+a^{-1}b^Tg^{-1}ba^{-1} & -a^{-1}b^Tg^{-1} \\ -g^{-1}ba^{-1} & g^{-1} \end{pmatrix},
\label{apbeq2r}
\end{equation}
 where $g=c-ba^{-1}b^T$. It follows that $ba^{-1}b^T \geq 0$ because $a^{-1}>0$ (see Eq.~(\ref{cdsd})) and therefore $c \geq g$, which implies that $c^{-1} \leq g^{-1}$. 

When $b=0$ (i.e., $M$ is block diagonal) then $c=g$ which implies that $c^{-1}=g^{-1}$. Now, 
\begin{equation}
\left[ba^{-1}b^T\right]_{kk}=\bs{b}(k)^Ta^{-1}\bs{b}(k),
\label{vec-mat-eq}
\end{equation}
 where $b^T=(\bs{b}(1),\bs{b}(2),\dots)$, i.e., we have written $b^T$ as a row vector of column vectors. As $a^{-1}>0$, and via Eq.~(\ref{vec-mat-eq}) and the definition of a positive definite matrix, then if $\bs{b}(k)\neq 0$ it follows that $[ba^{-1}b^T]_{kk}>0$. This implies that $ba^{-1}b^T=0$ only if $b=b^T=0$. Hence, because obviously $c \neq g$ if and only if $ba^{-1}b^T \neq 0$ then $c \neq g$ if and only if $b\neq 0$. Therefore, we have shown that the inverse of the  bottom right diagonal matrix in $M$, $c^{-1}$, is less than or equal to the bottom right diagonal matrix in $M^{-1}$ with the equality obtained only when $M$ is block-diagonal.

Now, by noting that $\tilde{A}(k)$ has been written in the form of the matrix in Eq.~(\ref{matrix-22}), and satisfies the conditions demanded of it ($\tilde{A}(k)>0$), we may then infer that
\begin{equation} 
  [\tilde{A}(k)^{-1}]_{\text{br}} \geq \left[ A_{[kk]} \right]^{-1} , \label{Eq-dsfs}
 \end{equation}
 where $[\tilde{A}(k)^{-1}]_{\text{br}}$ is the $d_k \times d_k$ sub-matrix of $\tilde{A}(k)^{-1}$ in the bottom right corner of $\tilde{A}(k)^{-1}$. Furthermore, the equality only holds when $A_k=0$, implying that $A_{[kj]}=0$ for all $j \neq k$, and as $A$ is symmetric this implies that $A_{[jk]}=0$ for all $j \neq k$. Now $\tilde{A}(k)^{-1}=P_{k} A^{-1} P^T_k$, which implies that $  [\tilde{A}(k)^{-1}]_{\text{br}}=[A^{-1}]_{[kk]} $. Hence, by putting this into Eq.~(\ref{Eq-dsfs}) this leads us to the final conclusion that
   \begin{equation} [A^{-1}]_{[kk]} \geq \left[ A_{[kk]} \right]^{-1} ,
 \end{equation}
with the equality obtained if and only if $A_{[jk]}=A_{[kj]}=0$ for all $j \neq k$. \end{proof}

As in Theorem 1 of the main text, consider a QSN with commuting parameter generators, with the $d$-dimensional vector to be estimated $\bs{\phi} = (\bs{\phi}_{[1]}, \dots, \bs{\phi}_{[s]})$ where $\bs{\phi}_{[k]}$ is a vector containing all of the parameters encoded into sensor $k$. Moreover, as in the main text, consider the sensor-separable initial state of the QSN\begin{equation}
\ket{\varphi} = \bigotimes_{k=1}^s \left(\sum_{\lambda_k} \|\langle \psi | \lambda_k \rangle\| \ket{\lambda_k} \right), 
\label{ap-varphi}
\end{equation}
where $\{\ket{\lambda_k}\}$ is a set of orthonormal mutual eigenstates of the generators for all of the parameters encoded into sensor $k$.

\begin{proposition} For any QSN and state $\varphi$ as described above, there exists a measurement of $\varphi_{\bs{\phi}} = U_{\bs{\phi}} \ket{\varphi}\bra{\varphi}U_{\bs{\phi}}^{\dagger}$ that saturates the QCRB and that can be implemented by local measurements on each of the sensors.
\end{proposition}
\begin{proof}
As noted in the main text, the QFIM for $\varphi$ is block-diagonal. In particular, all of the between-sensor terms are zero. Therefore an estimation procedure using this state can be treated as a collection of independent multi-parameter estimation problems: one at each sensor. Because the parameter generators all commute, there exists a POVM on sensor $k$ that saturates the QCRB for the vector encoded into this sensor (as Eq.~(1) of the main text is satisfied \cite{matsumoto2002new}). This POVM clearly need only act on this sensor (although to implement this POVM, ancillary systems may be required, e.g., a local optical reference beam). As there is an optimal POVM for estimating the parameters encoded at each sensor that is just a local operation on that sensor, these POVMs may all be applied in parallel to the entire sensing network. This is a measurement on the entire network that (a) saturates the QCRB, and (b) requires only local POVMs. 
\end{proof}

In the following proposition, we continue to consider the type of QSN described above. As in the main text, we denote our estimation uncertainty by $E_{\bs{\Phi}}$ (the definition of this quantity is given in the main text).
\begin{proposition} Consider a QSN with commuting generators and the Hilbert space $\mathcal{H}$. For any mixed state $\rho \in \mathscr{D}(\mathcal{H})$, there exists a pure state $\psi \in \mathscr{D}(\mathcal{H})$ such that (1) $\psi$ has an equal or lower QCRB on $E_{\bs{\Phi}}$  than $\rho$, and (2) the resources contained in $\psi$ and $\rho$ are equal.
\end{proposition}

\begin{proof} Any density operator $\rho \in \mathscr{D}(\mathcal{H})$ satisfies $\rho = \text{Tr}_{\mathbb{A}} (\ket{\Psi_\rho}\bra{\Psi_\rho})$ from some $\ket{\Psi_{\rho}} \in  \mathcal{H} \otimes  \mathcal{H}_{\mathbb{A}}$ and some ancillary Hilbert space $\mathcal{H}_{\mathbb{A}}$, with $ \mathcal{H}_{\mathbb{A}} = \mathcal{H}$ always sufficient \cite{nielsen2010quantum}. $\ket{\Psi_{\rho}}$ is known as a \emph{purification} of $\rho$. It is clear that $\mathcal{F}(\Psi_{\rho}) \geq \mathcal{F}(\rho)$, as one possible measurement strategy with the pure probe $\ket{\Psi_{\rho}}$ is to discard the ancillary sensor(s), which is entirely equivalent to having the probe state $\rho$.  Now any such purification $\ket{\Psi_{\rho}}$ has an equal or worse QCRB on $E_{\bs{\Phi}}$ than the state $\ket{\varphi'} = \ket{\varphi} \otimes \ket{\psi_{\mathbb{A}}}$, where $\ket{\varphi} $ is the separable state constructed in the proof of Theorem 1 in the main text (see also Eq.~\ref{ap-varphi}) -- the precise form of which will be dependent on $\Psi_{\rho}$ -- and $\ket{\psi_{\mathbb{A}}}$ is any state of the ancillary system(s). Moreover, we may simply drop the ancillary systems in the state $\ket{\varphi'}$, as they do not affect the QFIM of $\ket{\varphi'}$. As such, $\varphi = \ket{\varphi}\bra{\varphi}$ has a smaller QCRB on $E_{\bs{\Phi}}$ than does $\rho$, for any $\rho$ (noting that $\varphi$ depends on $\rho$). Finally, by construction $\rho$ and $\varphi$ contain the same amount of resources. Hence, $\varphi$ satisfies the required conditions of $\psi$ in this proposition.
\end{proof}


\section{Proof of Theorem 2 \label{Sec:proof2}}
In this section, we prove Theorem 2 of the main text. This is restated here for convenience. To be clear, in the following theorem we are consider a general QSN, where (1) the aim to estimate the vector $\bs{\phi} =(\bs{\phi}_{[1]},\dots,\bs{\phi}_{[s]})$, where $\bs{\phi}_{[k]}$ is encoded into the $k^{\rm th}$ sensor; (2) the generators are not assumed to all commute (in contrast to Theorem 1).

\begin{theorem2}
Consider any QSN in which we wish to minimize $E_{\bs{\Phi}}$. For any estimator, probe state $\rho \in \mathscr{D}(\mathcal{H})$ and measurement $\mathcal{M}_{\rho} \in \mathscr{M}(\mathcal{H})$, there exists an estimator, probe $\varphi \in \mathscr{D}(\mathcal{H} \otimes \mathcal{H})$ and measurement $\mathcal{M}_{\varphi} \in \mathscr{M}(\mathcal{H} \otimes \mathcal{H})$ for which 
\begin{enumerate}
\item $\varphi$ is separable between sensors, but each sensors can be entangled with a local ancilla.
\item $R(\varphi) \leq 2R(\rho)$.
 \item $\mathcal{M}_{\varphi}$ is implementable by independent measurements of each sensor.
\item $E_{\bs{\Phi}}(\varphi, \mathcal{M}_{\varphi}) \leq E_{\bs{\Phi}}(\rho, \mathcal{M}_{\rho})$ in the asymptotic $\mu$ limit.
\end{enumerate}
\end{theorem2}

\begin{proof}
The proof is split into three parts: (i) We show that for any $\rho$ we can construct a pure state $\varphi \in \mathscr{D}(\mathcal{H}\otimes\mathcal{H})$ that satisfies condition 1 and that has an equal or better QCRB on $E_{\bs{\Phi}}$ than does $\rho$; (ii) We show that there is a measurement on $\varphi$ that satisfies condition 3 and 4; (iii) We prove that $\varphi$ satisfies condition 2.

\vspace{0.1cm}
\emph{Part (i)} -- For entirely general generators, the elements of the QFIM for a pure probe state are given by \cite{liu2015quantum,liu2014fidelity} 
\begin{equation} 
\mathcal{F}_{mn}(\psi) = 2 \langle  \{\hat{H}_m,\hat{H}_n \} \rangle  -4 \langle \hat{H}_m \rangle \langle \hat{H}_n  \rangle,
\label{eq:QFIM-non-com}
 \end{equation} 
where $\{\cdot,\cdot\}$ is the anti-commutator ($\{A,B\}=AB+BA$). In our QSN problem, the generators of parameters encoded into different sensors must commute. That is, $[ \hat{H}_k,\hat{H}_l] =0$ for $k \in \mathbb{P}_p$ and $l \in \mathbb{P}_q$ with $p \neq q$, where we are using the notation introduced in Eq.~\eqref{eq:pj}. More importantly, this also implies that $\hat{H}_k$ acts non-trivially only in sub-space $\mathcal{H}_l$ if $k\in\mathbb{P}_l$ (where $\mathcal{H}_l$ is the Hilbert space of sensor $l$). 

Now consider any probe state $\rho \in \mathscr{D}(\mathcal{H})$, and any purification of $\rho$ into the Hilbert space  $\mathcal{H} \otimes \mathcal{H}$, which we denote $\psi_{\rho} \in \mathscr{D}(\mathcal{H}\otimes \mathcal{H})$ ($\rho$ can always be purified into this duplicated Hilbert space \cite{nielsen2010quantum}). This purified state must have an optimal estimation uncertainty (i.e., $E_{\bs{\Phi}}$ minimized over all measurements) that is equal to or smaller than that of $\rho$. This is because any measurement strategy for $\rho$ is equivalent to one for $\psi_{\rho}$ where the additional sensors are discarded. 

Denote the QFIM of $\psi_{\rho}$ by $\mathcal{F}$. For a pure state, the $\mathcal{F}_{[ll]}$ sub-matrix of $\mathcal{F}$ depends only on the reduced density operator 
 \begin{equation}
 \rho_{l}=\text{Tr}_{\mathbb{S} \setminus l }(\psi_\rho), 
 \end{equation}
 where $\mathbb{S} \setminus l$ denotes the set of all the sensors except sensor $l$. This follows from Eq.~(\ref{eq:QFIM-non-com}), and by noting that $\hat{H}_k$ acts non-trivially only in the Hilbert space on which $\phi_k$ is encoded. But we can also find a pure state in $\ket{\varphi_l} \in \mathcal{H}_{l} \otimes \mathcal{H}_{l}$ with the same reduced density matrix, $\rho_l$, obtained by tracing over the second ancillary sensor. Therefore the state
 \begin{equation}
 \ket{\varphi} = \ket{\varphi_1} \otimes \ket{\varphi_2} \otimes \dots \otimes | \varphi_{s}\rangle,
 \end{equation}
where $\ket{\varphi_k}$ is a purification of $\rho_{k}$ for $k=1,\dots,s$, has a QFIM $\mathcal{F}'$ with $\mathcal{F}'_{[ll]}=\mathcal{F}_{[ll]}$ for all $l$. Now, as $\varphi = \ket{\varphi}\bra{\varphi}$ contains no entanglement between any two sensors (but note that $\varphi$ does generally contain entanglement between a sensor and its local ancillary duplicate), the off-diagonal sub-matrices of $\mathcal{F}'$ are zero. This implies that
\begin{equation}
 \mathcal{F}'=  \begin{pmatrix}
  \mathcal{F}_{[11]} &0 & \cdots & 0 \\
 0& \mathcal{F}_{[22]} & \cdots & 0 \\
  \vdots  & \vdots  & \ddots & \vdots  \\
  0& 0& \cdots & \mathcal{F}_{[mm]}
 \end{pmatrix}.
\end{equation}
Now, from the inequality in Eq.~(\ref{eq:qfim-in-}), it follows that
\begin{equation}
  \frac{1}{\mu}\text{Tr}( W \mathcal{F}^{-1}) \geq\frac{1}{\mu}\text{Tr} ( W  \mathcal{F}'^{-1}) ,
 \label{Eq:diag-f-bound2}
  \end{equation}
  for any weighting matrix $W$, with the LHS of this inequality the QCRB on $E_{\bs{\Phi}}$ with the probe state $\psi_{\rho}$, and the RHS of this inequality the QCRB on $E_{\bs{\Phi}}$ with the probe state $\varphi$. Hence, $\varphi$ has an equal or lower QCRB bound on the estimation uncertainty than the purified state $\psi_\rho$, for any $\rho$ and any purification. Therefore, $\varphi$ also has a lower QCRB bound on $E_{\bs{\Phi}}$ than that for $\rho$, and note that $\varphi$ satisfies condition 1 of the theorem.

\vspace{0.1cm}
 \emph{Part (ii) --}
Because we do not know that the QCRB can be saturated (for non-commuting generators it often cannot be saturated) simply showing that the QCRB bound on $E_{\bs{\Phi}}$ for the separable state $\varphi$ is smaller or equal to the bound on $E_{\bs{\Phi}}$ for $\rho$, for any $\rho$, is insufficient to show that $\varphi$ necessarily has a better estimation precision than $\rho$ when an optimal measurement is chosen. However, we can confirm this is the case, with the following argument. 

The precision with which $\phi_{[l]}$ can be measured is always improved or unaffected if we know $\phi_{[k]}$ for all $k \neq l$. Both $\psi_{\rho}$ and $\varphi$ have the same QFIM for $\phi_{[l]}$, which is $\mathcal{F}_{[ll]}$, and if all the other parameters are known we may set them to zero (by local known unitaries before the measurement). Therefore, the $\phi_{[l]}$-encoded state in each of these cases is 
\begin{align} 
\ket{\psi_{\rho}^l} &\equiv (\mathds{1}   \otimes \dots \otimes U_l(\phi_{[l]})    \otimes \dots \otimes \mathds{1} ) \ket{\psi_{\rho}}, \\
\ket{\varphi^l} &\equiv    \ket{\varphi_{1}} \otimes \dots \otimes  \ket{\varphi_{l}^l} \otimes \dots \otimes| \varphi_{s}\rangle,
\end{align}
where $\ket{\varphi_{l}^l} \equiv (U_l(\phi_{[l]}) \otimes \mathds{1}) \ket{\varphi_{l}}$. Using only $\phi_{[l]}$-\emph{independent} unitary operations and partial traces (on an extended Hilbert space), we may map $\ket{\varphi^l} \to \ket{\psi_{\rho}^l}$. We relegate a proof of this to the following section. Hence, any POVM on $\ket{\psi_{\rho}^l}$ is exactly equivalent to some POVM on $\ket{\varphi^l}$. This is in the sense that the POVMs have the same number of POVM effects and each measurement outcome, $m$, is associated with the same probability density function, $p(m|\phi_{[l]}$). This implies that $\ket{\varphi^l}$ can estimate $\phi_{[l]}$ with at least as small an estimation uncertainty as can be obtained with $\ket{\psi_{\rho}^l}$, when all the other parameters are known and if the optimal measurement is used. Note that this measurement might not saturate the QCRB bound for $E_{\Phi_{[l]}}$, and when this is the case the optimal measurement will depend on the weighting sub-matrix $W_{[ll]}$.

Return now to the actual problem of interest -- when all of the parameters are \emph{not} known. For the separable state $\varphi$, all of the $\phi_{[l]}$ can be measured simultaneously to the same, or a better, precision that $\psi_{\rho}$ can estimate each $\phi_{[l]}$ when all of the other $\phi_{[k]}$ \emph{are} known. This is because $\varphi$ is separable between any two sensor-and-local-ancilla pairs, and so the optimal POVM for estimating $\phi_{[l]}$ with that state, and given $W_{[ll]}$, need only act on the $l$\textsuperscript{th} duplicate sensor (for exactly the reasons given in the proof of Proposition 2, except that now the QCRB possibly cannot be saturated). Hence, all of the measurements to optimize the estimation precision of each $\phi_{[l]}$ can be implemented in parallel. However, there is no guarantee that $\psi_{\rho}$ can estimate all of the $\phi_{[l]}$ simultaneously with the same estimation uncertainty that each one can be estimated with when all of the other parameters are known. 

Hence, we have shown that we can map any density operator $\rho$ to a pure state $\varphi$ that is separable between sensors (but may have entanglement between a sensor and a local ancilla), and that, for some POVM that can be implemented with independent local measurements, the estimation uncertainty obtained with $\varphi$ is equal or lower than that obtained with $\rho$, for any measurement on $\rho$. Thus, although the QCRB cannot necessarily be saturated, a separable state allows us to get as close as it is possible to saturating it. As such, we have now found a state and measurement for which conditions 1, 3 and 4 hold.

\vspace{0.1cm}   
\emph{Part (iii) --}
Finally, we need to show that the resources consumed by $\varphi$ are at most twice those in $\rho$. In order to assess the resources contained in a state of $s$ sensors and $s$ ancillas it is necessary to define how to count resources in the ancillas. The natural ``worst case'' extension of the resource operator $\hat{R}= \hat{R}_1 + \hat{R}_2 + \cdots + \hat{R}_{s}$ (see main text) to a larger sensors-and-ancillas Hilbert space is of the form $\hat{R} \to \hat{R}'$ with
\begin{equation} 
\hat{R}'= (\hat{R}_1 +\hat{R}_1') + (\hat{R}_2 +\hat{R}_2') +\dots + (\hat{R}_{s} +\hat{R}_{s}'),
\end{equation}
where $\hat{R}_k'$ is the same as the operator $\hat{R}_k$ except that it instead acts non-trivially only on the ancillary sensor local to sensor $k$ (whereas $\hat{R}_k$ acts non-trivially only on sensor $k$). This is because with this choice for $\hat{R}'$ we are counting resources in the ancillas on an equal footing to resources consumed by the sensors. Once this is understood, it is immediately clear that $\varphi$ need contain no more than twice the amount of resources as $\rho$. The exact amount depends on the chosen purification to construct $\varphi$, and could be considerably less than this, but it need never be greater than this.
\end{proof}

Note that in the proof of Theorem 2, we took the most conservative approach to resource counting with ancillas that seems physically sensible. An alternative well-motivated choice for the resource operators is to take the resource operator on each ancilla to be the operator that maps all vectors to zero. This is the relevant choice when all properties of the ancillary systems are irrelevant from the perspective of resource counting. This is arguably the most appropriate method for counting resources when the parameters are induced by some fragile sample (relevant optical sensing examples include measurements of spin ensembles \cite{wolfgramm2013entanglement}, biological systems \cite{carlton2010fast,taylor2013biological}, atoms \cite{tey2008strong,eckert2008quantum} and single molecules \cite{pototschnig2011controlling}). In this case, it is essential to minimize the disturbance of the sample, and as any ancillary systems do not interact with the sample there is no need to minimize any property (e.g., energy) local to that part of the state. In this setting, our argument implies that there is no fundamental improvement gained from using sensor-entangled states, as we noted in the main text.

\section{Equivalent POVMs}\label{app:POVMs}
In this final section, we confirm the claim made in the proof of Theorem 2: using only $\phi_{[l]}$-\emph{independent} unitary operations and partial traces (on an extended Hilbert space), we may map $\ket{\varphi^l} \to \ket{\psi_{\rho}^l}$. The precise forms for these states were given in the proof of Theorem 2, and will be repeated later in this section. To prove this claim we show something more general, and then show how it applies in this particular case.

Consider a density operator, $\rho$, on some Hilbert space, $\mathcal{H}$, with dimension $q$. Now consider any purification of $\rho$ into $\mathcal{H} \otimes \mathcal{H} $, denoted $\ket{\Psi_1}$, and another purification of $\rho$ into a Hilbert space $\ket{\Psi_2} \in \mathcal{H} \otimes \mathcal{H}'$, where $\mathcal{H}'$ is of dimension $q'\geq q$. Consider the states obtained by enacting the local unitary $u$ on the `original' Hilbert space, i.e., the states 
\begin{align}
\ket{\Psi_1(u)} &= (u \otimes \mathds{1}_{q} ) \ket{\Psi_1},\\
\ket{\Psi_2(u)} &= (u \otimes \mathds{1}_{q'}) \ket{\Psi_2}.
\end{align}
Via only $u$-independent unitary transformations and partial traces, we may map $\ket{\Psi_1(u)} \otimes \ket{\text{fid}'} \to \ket{\Psi_2(u)} $, where $\ket{\text{fid}'}$ is some fiducial state in $\mathcal{H}'$. 

\begin{proof}
 It is always possible to express $\ket{\Psi_1(u)}$ as
\begin{align} 
\ket{\Psi_1(u)} = \sum_{k=1}^q \alpha_k  \ket{\gamma_k^u} \otimes \ket{\varphi_k}, 
\end{align}
where the $\ket{\gamma_k^u}$ and $\ket{\varphi_k}$ states form orthonormal bases for $\mathcal{H}$, and only the $\ket{\gamma_k^u}$ depend on $u$. Because $\ket{\Psi_2(u)}$ is also a purification of $\rho$ it must be possible to express it in the similar form
\begin{align} 
 \ket{\Psi_2(u)} = \sum_{k=1}^q \alpha_k  \ket{\gamma_k^u} \otimes \ket{\vartheta_k}, 
\end{align}
where the $\ket{\vartheta_k}$ are $q$ states from an orthonormal basis of $\mathcal{H}'$ (that is, $\ket{\vartheta_k}$ for $k=1,\dots,q'$ is an orthonormal basis for $\mathcal{H}'$).

Now consider any unitaries, $U_k'$, such that $U_k' \ket{\text{fid}'} =  \ket{\vartheta_k}$ for $k=1,\dots,q$ (note that this relation does not fully define any of the unitaries). Using any such unitaries, we may construct the unitary
\begin{align} 
 \Lambda_A  =\mathds{1} \otimes  \sum_{k=1}^q \ket{\varphi_k} \bra{\varphi_k} \otimes U_k' ,
\end{align}
which acts on $\mathcal{H}_{T} = \mathcal{H} \otimes \mathcal{H} \otimes \mathcal{H}'$. For any such $\Lambda_A$ it follows that
\begin{align} 
 \Lambda_A( \ket{\Psi_1(u)} \otimes\ket{\text{fid}'}) = \sum_{k=1}^q\alpha_k  \ket{\gamma_k^u} \otimes \ket{\varphi_k} \otimes \ket{\vartheta_k}.   
\end{align}
In essentially the same fashion we have that
\begin{align*} 
\Lambda_B\Lambda_A( \ket{\Psi_1(u)} \otimes\ket{\text{fid}'}) = \sum_{k=1}^q \alpha_k  \ket{\gamma_k^u} \otimes \ket{\text{fid}} \otimes \ket{\vartheta_k},
\end{align*}
where $ \Lambda_B$ is a unitary on $\mathcal{H}_{T}$ defined by
\begin{align} 
 \Lambda_B  =\mathds{1} \otimes   \sum_{k=1}^{q'} U_k^{\dagger} \otimes \ket{\vartheta_k} \bra{\vartheta_k} , 
\end{align}
where $U_k$ are any unitaries with the action $U_k \ket{\text{fid}} = \ket{\varphi_k}$ for $k=1,\dots,q$, where $ \ket{\text{fid}}$ is some fixed state in $\mathcal{H}$, and $U_k$ may have any action for $k=q+1,\dots, q'$. Therefore, denoting
\begin{align} 
\ket{\xi(u)} = \Lambda_B\Lambda_A (\ket{\Psi_1(u)} \otimes \ket{\text{fid}'} ),
\end{align}
we have that 
\begin{align} 
 \ket{\Psi_2(u)} \bra{\Psi_2(u)} = \text{Tr}_{2} \left(\ket{\xi(u)} \bra{\xi(u)}\right),
 \end{align}
where the trace operation is over the second Hilbert space in $\mathcal{H}_{T}=\mathcal{H} \otimes \mathcal{H} \otimes \mathcal{H}'$. Hence, we can map $\ket{\Psi_1(u)}$ to $\ket{\Psi_2(u)}$ using only $u$-independent unitary transformations and a partial trace.
\end{proof}

As stated at the beginning of this section, in the proof of Theorem 2 we considered the two $\phi_{[l]}$-encoded states
\begin{align} 
\ket{\psi_{\rho}^l} &=(\mathds{1}   \otimes \dots \otimes U_l(\phi_{[l]})    \otimes \dots \otimes \mathds{1} ) \ket{\psi_{\rho}}, \\
\ket{\varphi^l} &=    \ket{\varphi_{1}} \otimes \dots \otimes  \ket{\varphi_{l}^l} \otimes \dots \otimes \ket{\varphi_{s}},
\end{align}
where $\ket{\varphi_{l}^l}=(U_l(\phi_{[l]}) \otimes \mathds{1}) \ket{\varphi_{l}}$, and we claimed that using only $\phi_{[l]}$-independent unitary operations and partial traces (on an extended Hilbert space) we may map $\ket{\varphi^l} \to \ket{\psi_{\rho}^l}$. We may clearly map $\ket{\varphi^l}\to \ket{\varphi_{l}^l} $ using a partial trace, so we only need to show that the result above implies that we may  $\ket{\varphi_{l}^l} \to \ket{\psi_{\rho}^l}$ using only $\phi_{[l]}$-independent unitary operations and partial traces. Both $\ket{\varphi_l}$ and $\ket{\psi_{\rho}}$ are purifications of the same density operator $ \rho_{l}=\text{Tr}_{\mathbb{S} \setminus l }(\rho)$ on sensor $l$. In particular, $\ket{\varphi_l}$ is a purification into the doubled Hilbert space and $\ket{\psi_{\rho}}$ a purification into a larger Hilbert space. Furthermore, $\ket{\varphi_{l}^l} $ and $\ket{\psi_{\rho}^l}$ are simply $\ket{\varphi_l}$ and $\ket{\psi_{\rho}}$, respectively, evolved by some $\phi_{[l]}$-dependent unitary that is local to the `original' Hilbert space. As such, it is clear that our derivation above implies that there is a mapping $\ket{\varphi_{l}^l} \to \ket{\psi_{\rho}^l}$ which uses only $\phi_{[l]}$-independent unitary operations and partial traces (on an extended Hilbert space).

\end{document}